\documentclass[a4paper,UKenglish,cleveref,thm-restate]{lipics-v2021}

\usepackage{xspace}
\usepackage{mathtools}
\usepackage{complexity}
\usepackage{csquotes}

\renewcommand{\R}{\ensuremath{\mathbb{R}}\xspace}
\newcommand{\Z}{\ensuremath{\mathbb{Z}}\xspace}
\newcommand{\Q}{\ensuremath{\mathbb{Q}}\xspace}
\newcommand{\N}{\ensuremath{\mathbb{N}}\xspace}

\newcommand{\ER}{\ensuremath{\exists\R}\xspace}
\newcommand{\UR}{\ensuremath{\forall\R}\xspace}
\newcommand{\UER}[1][]{{\normalfont\ensuremath{\forall\exists_{#1}\R}}\xspace}
\newcommand{\EUR}[1][]{{\normalfont\ensuremath{\exists\forall_{#1}\R}}\xspace}

\newcommand{\QFF}[1][]{{\normalfont\ensuremath{\mathrm{QFF_{#1}}}}\xspace}
\newcommand{\eps}{\varepsilon}
\newcommand{\til}{\widetilde}
\newcommand{\dH}{d_{\mathrm{H}}}
\newcommand{\dDH}{\vec{d}_{\mathrm{H}}}
\DeclarePairedDelimiter\abs{\lvert}{\rvert}
\DeclarePairedDelimiter\norm{\lVert}{\rVert}
\newcommand{\sepdot}{\, . \,}
\newcommand{\counterexample}[1]{\bot(#1)}
\newcommand{\dequiv}{:\equiv}
\newcommand{\wequiv}{\ \equiv\ }
\newcommand{\wdequiv}{\ :\equiv\ }

\newcommand{\problemname}[1]{{\normalfont\textsc{#1}}}
\newcommand{\Hausdorff}{\problemname{Hausdorff}\xspace}
\newcommand{\DirectedHausdorff}{\problemname{Directed Hausdorff}\xspace}

\newcommand{\UETR}{\problemname{UETR}\xspace}
\newcommand{\StrictUETR}{\problemname{Strict-UETR}\xspace}
\newcommand{\UStrict}{$\forall$-strict\xspace}
\newcommand{\ExoticUETR}{\problemname{Exotic-UETR}\xspace}

\newtheorem{openproblem}{Open Problem}

\title{The Complexity of the Hausdorff Distance}
\titlerunning{The Complexity of the Hausdorff Distance}

\author{Paul Jungeblut}{Karlsruhe Institute of Technology, Germany}{paul.jungeblut@kit.edu}{https://orcid.org/0000-0001-8241-2102}{}

\author{Linda Kleist}{Technische Universität Braunschweig, Germany}{kleist@ibr.cs.tu-bs.de}{https://orcid.org/0000-0002-3786-916X}{}

\author{Tillmann Miltzow}{Utrecht University, The Netherlands}{t.miltzow@uu.nl}{https://orcid.org/0000-0003-4563-2864}{}

\authorrunning{P. Jungeblut, L. Kleist and T. Miltzow}

\Copyright{Paul Jungeblut, Linda Kleist and Tillmann Miltzow}

\ccsdesc[100]{Theory of computation $\rightarrow$ Computational geometry}

\keywords{%
    Hausdorff Distance,
    Semi-Algebraic Set,
    Existential Theory of the Reals,
    Universal Existential Theory of the Reals,
    Complexity Theory
}

\funding{
    Linda Kleist was partially supported by a postdoc fellowship of the German Academic Exchange Service (DAAD).
    Tillmann Miltzow is generously supported by the Netherlands Organisation for Scientific Research (NWO) under project no. 016.Veni.192.250.
}

\acknowledgements{
    We thank anonymous reviewers for valuable feedback for earlier versions of this manuscript.
    We thank Kristoffer Hansen for pointers to the literature.
}

\hideLIPIcs
\nolinenumbers

\begin{document}

\maketitle

\begin{abstract}
  We investigate the computational complexity of computing the Hausdorff distance.
  Specifically, we show that the decision problem of whether the Hausdorff distance of  two semi-algebraic sets is bounded by  a given threshold is complete for the complexity class \UER[<].
  This implies that the problem is \NP-, \co-\NP-, \ER- and \UR-hard.
\end{abstract}

\section{Introduction}

The question of \enquote{how similar are two given objects} occurs in numerous settings. For three concrete examples, consider \cref{fig:hausdorff_examples}.

\begin{figure}[htb]
    \centering
    \begin{subfigure}[b]{0.24\textwidth}
        \centering
        \includegraphics[page=1]{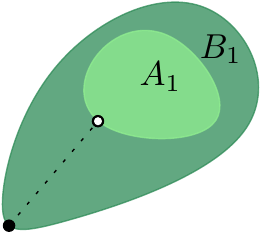}
        \caption{$\dH(A_1,B_1)$}
        \label{fig:hausdorff_examples_1}
    \end{subfigure}
    \hfill
    \begin{subfigure}[b]{0.3\textwidth}
        \centering
        \includegraphics[page=2]{figures/hausdorff-examples.pdf}
        \caption{$\dH(A_2,B_2)$}
        \label{fig:hausdorff_examples_2}
    \end{subfigure}
    \hfill
    \begin{subfigure}[b]{0.2\textwidth}
        \centering
        \includegraphics[page=3]{figures/hausdorff-examples.pdf}
        \caption{$\dH(A_3,B_3)$}
        \label{fig:hausdorff_examples_3}
    \end{subfigure}
    \caption{How similar are these sets?}
    \label{fig:hausdorff_examples}
\end{figure} 

A typical tool to quantify their similarity is the Hausdorff distance.
Two sets have a small Hausdorff distance if every point of one set is close to some point of the other set and vice versa.
It is well known that the Hausdorff distance appears in many branches of science.
To illustrate the range of use cases, we consider two examples, for illustrations see \cref{fig:motivation}.
In mathematics, the Hausdorff distance provides a metric on sets and henceforth also a topology.
This topology can be used to discuss continuous transformations of one set to another~\cite{Bredon2013_Topology}.
In computer vision and geographical information science, the Hausdorff distance is used to measure the similarity between spacial objects~\cite{Min2007_HausdorffInGIS,Rucklidge1996_Hausdorff}, for example the quality of quadrangulations of complex 3D models~\cite{Verhoeven2022_Dev2PQ}.
In this paper, we study the computational complexity of the Hausdorff distance from a theoretical perspective.

\begin{figure}[htb]
    \centering
    \begin{subfigure}[b]{0.47\textwidth}
        \centering
        \includegraphics[page=1]{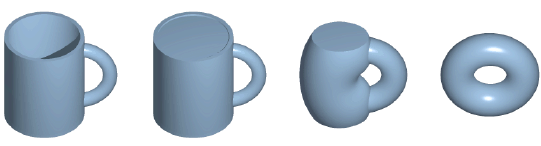}
        \caption{Continuous deformation of a cup into a doughnut~\cite{Wikipedia_Homotopy}.}
        \label{fig:motivation_cups}
    \end{subfigure}
    \hfill
    \begin{subfigure}[b]{0.45\textwidth}
        \centering
        \includegraphics[page=2]{figures/motivating-examples.pdf}
        \caption{Quadrangulation of a smooth surface used for rendering~\cite{Verhoeven2022_Dev2PQ}.}
        \label{fig:motivation_quadragular_mesh}
    \end{subfigure}
    \caption{The Hausdorff distance in different branches of science.}
    \label{fig:motivation}
\end{figure}

\subsection{Problem Definition}
The \emph{directed Hausdorff distance} between two non-empty sets $A,B \subseteq \R^n$ is defined as
\[
    \dDH(A,B) :=
    \adjustlimits \sup_{a \in A} \inf_{b \in B} \norm{a - b}
    \text{.}
\]
The directed Hausdorff distance between~$A$ and~$B$ can be interpreted as the smallest value $t \geq 0$ such that the (closed) $t$-neighborhood of~$B$ contains~$A$.
Hence, it nicely captures the intuition of how much~$B$ has to be expanded uniformly in all directions to contain~$A$.
Note that this definition is not symmetric, so $\dDH(A,B)$ and $\dDH(B,A)$ may differ. 
For an example, consider \cref{fig:hausdorff_examples_1};
while $A_1 \subseteq B_1$ and thus $\dDH(A_1,B_1) = 0$, it holds that $\dDH(B_1,A_1) > 0$.
In contrast, the (undirected) \emph{Hausdorff distance} is symmetric and defined as
\[
    \dH(A,B) := \max\bigl\{
        \dDH(A,B),
        \dDH(B,A)
    \bigr\}
    \text{.}
\]
In this paper, we investigate the computational complexity of deciding whether the Hausdorff distance of two sets is at most a given threshold.

\subsection{Semi-Algebraic Sets}
The algorithmic complexity of computing the Hausdorff distance clearly depends on the type of the underlying sets:
If both sets consist of finitely many points, their Hausdorff distance can be easily computed in polynomial time by checking all pairs of points.
However in practice, one often considers infinite sets such as collections of disks in the plane, cubic splines or surfaces in three (or more) dimensions, see also \cref{fig:semi_algebraic_set_examples}.

\begin{figure}[htb]
    \centering
    \begin{subfigure}[t]{0.27\textwidth}
        \centering
        \includegraphics[page=1]{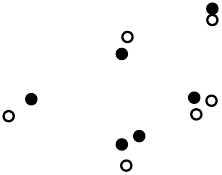}
        \caption{Two sets of white and black points in~$\R^2$.}
        \label{fig:semi_algebraic_set_points}
    \end{subfigure}
    \hfill
    \begin{subfigure}[t]{0.26\textwidth}
        \centering
        \includegraphics[page=2]{figures/settings-examples.pdf}
        \caption{Two sets of blue and red disks in $\R^2$.}
        \label{fig:semi_algebraic_set_disks}
    \end{subfigure}
    \hfill
    \begin{subfigure}[t]{0.37\textwidth}
        \centering
        \includegraphics[page=3]{figures/settings-examples.pdf}
        \caption{Two surfaces with different meshes in $\R^3$, image from~\cite{Verhoeven2022_Dev2PQ}.}
        \label{fig:semi_algebraic_set_surfaces}
    \end{subfigure}
    \caption{The Hausdorff distance in simple and more complicated settings.}
    \label{fig:semi_algebraic_set_examples}
\end{figure}

In this paper, we focus on semi-algebraic subsets of~$\R^n$ that can be described by polynomial inequalities with integer coefficients.
For simplicity, we use the term \emph{semi-algebraic set} in this paper.
Formally, a semi-algebraic set is the finite union of basic semi-algebraic sets.
A \emph{basic semi-algebraic set}~$S$ is specified by two families of polynomials~$\mathcal{P}$ and~$\mathcal{Q}$ with integer coefficients such that
\[
    S = \bigl\{
        x \in \R^n \mid
        \bigwedge_{P \in \mathcal{P}} P(x) \leq 0 \land
        \bigwedge_{Q \in \mathcal{Q}} Q(x) < 0
    \bigr\}
    \text{.}
\]
Semi-algebraic sets cover clearly the vast majority of practical cases.
Simultaneously, even in supposedly simple cases, i.e., when considering circles, ellipses or cubic splines, one has to use polynomial equations to describe the sets.

\subsection{Concrete Example}
In order to demonstrate how difficult it is in practice to compute the Hausdorff distance even between two curves in~$\R^2$, let us consider the following example (given by Bernd Sturmfels at a workshop in Saarbr\"{u}cken in 2019).
The two polynomials
\begin{align*}
    f(x,y) &:= x^4 + y^4 + 12x^3 + 2y^3 - 3xy + 11 \quad \text{and} \\
    g(x,y) &:= 7x^4 + 8y^4 - 1
\end{align*}
define sets $A = \{(x,y) \in \R^2 \mid f(x,y) = 0\}$ and $B = \{(x,y) \in \R^2 \mid g(x,y) = 0\}$.
For an illustration of~$A$ and~$B$, consider the blue and green curve in \cref{fig:concrete_example}, respectively.

\begin{figure}[htb]
    \centering
    \includegraphics{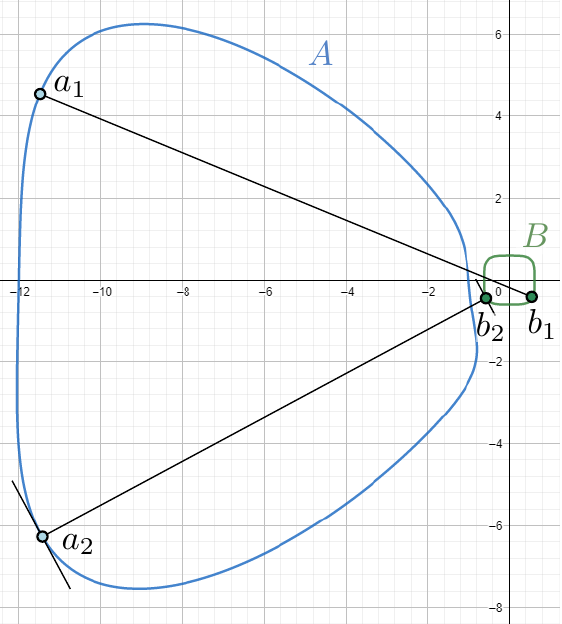}
    \caption{
        The Hausdorff distance between the compact semi-algebraic sets (in blue and green) is attained at points $(a_2, b_2)$ such that the segment $a_2b_2$ is orthogonal to the tangents at~$a_2$ and~$b_2$.
        While the segment $a_1b_1$ is longer than $a_2b_2$, the pair $(a_1, b_1)$ does not realize the Hausdorff distance because the segment $a_1b_1$ crosses both~$A$ and~$B$.
    }
    \label{fig:concrete_example}
\end{figure}

It can be argued using convexity and continuity that the Hausdorff distance is attained at points $a \in A$ and $b \in B$ such that the segment~$ab$ is orthogonal to the tangents at~$a$ and~$b$.
This yields a set of polynomial equations in four variables.
The system has 240 complex solutions, eight of which are real.
These 240 solutions can be computed using computer algebra systems based on Gr\"{o}bner bases.
For some real solutions~$(a,b)$, the segment~$ab$ crosses~$A$ and~$B$, for example~$a_1b_1$ as in \cref{fig:concrete_example}.
These solutions can be discarded.
Among the remaining solutions the points $a_2 \approx (-11.48362, -6.1760)$ and $b_2 \approx (-0.56460, -0.43583)$ realize the Hausdorff distance of approximately $12.33591$.
This approach does not easily generalize to general semi-algebraic sets.
In the next paragraph, we present a slower, but more general method.

\subsection{General Decision Algorithm}
Consider a situation where we are given two semi-algebraic sets~$A$ and~$B$ as well as a threshold~$t \in \N$.
The statement~$\dDH(A,B) \leq t$ can be encoded into a logical sentence of the form
\begin{equation}
    \label{eqn:directed_hausdorff}
    \forall \eps > 0, a \in A \sepdot
    \exists b \in B :
    \norm{a-b}^2 \leq t^2 + \eps
    \text{,}
\end{equation}
where~$\eps$ is needed to also consider points in the closure of~$B$.
We can decide the truth of this sentence by employing sophisticated algorithms from real algebraic geometry that can deal with two blocks of quantifiers~\cite[Chapter 14]{Basu2006_RealAlgebraicGeometry}.
These algorithms are so slow that they are impractical in practice, even for small instances like the above example.
Our main result roughly states that in general there is little hope for an improvement.
To state this formally, we continue by defining suitable complexity classes.

\subsection{Algorithmic Complexity}
Let~$\varphi$ be a \emph{quantifier-free formula} in the \emph{first-order theory of the reals} with free variables $X = (X_1, \ldots, X_n)$ and $Y = (Y_1, \ldots, Y_m)$.
See \cref{sec:fotr} for a formal definition of the syntax and semantics.
For now, think of a set of polynomial equations and inequalities, called \emph{atoms}, with integer coefficients in the variables~$X$ and~$Y$.
These atoms are combined into formula~$\varphi$ using the logical connectives $\land$, $\lor$ and~$\neg$ (also, parenthesis are allowed).
The \problemname{Universal Existential Theory of the Reals} (\UETR) asks to decide whether a sentence of the form
\[
    \forall X \in \R^n \sepdot
    \exists Y \in \R^m :
    \varphi(X,Y)
\]
is true.
We define the following restriction of \UETR:
If~$\varphi$ does not contain negations~(no $\neg$) and all atoms are strict inequalities (only $<$, $>$ or $\neq$), then we denote the corresponding decision problem by \StrictUETR.
Of course, \StrictUETR is a special case of \UETR, so it is at most as difficult.

We capture this by defining the complexity classes \UER and \UER[<] to contain all decision problems that polynomial-time many-one reduce to \UETR and \StrictUETR, respectively.
We propose to pronounce the complexity class \UER as \enquote{UER} or \enquote{forall exists R} and \UER[<] as \enquote{strict-UER}  or \enquote{strict forall exists R}.
Let us emphasize that we work in the bit-model of computation; all inputs have finite precision and their overall length determines the size of the problem instance.
To the best of our knowledge, \UER was first introduced by B\"{u}rgisser and Cucker~\cite[Section~$9$]{Burgisser2009_ExoticQuantifiers} under the name $\text{\textsf{BP}}^0(\forall\exists)$ (in the constant-free Boolean part of the Blum-Shub-Smale-model~\cite{Blum1989_ComputationOverTheReals}).
The notation \UER arised later in~\cite{Dobbins2018_AreaUniversality} extending the notation from Schaefer and \v{S}tevankovi\v{c}~\cite{Schaefer2017_FixedPointsNash}.
The class \co\UER[<] = \EUR[\leq] was first studied by D'Costa, Lefaucheux, Neumann, Ouaknine and Worrel~\cite{DCosta2021_EscapeProblem}.

Concerning the relation of these complexity classes,  \UER[<] is contained in \UER because \StrictUETR is a special case of \UETR.
It is an intriguing open problem if those two classes coincide or are different, see \cref{sec:open_problems} for a short discussion.

\subsection{Problem and Results}
We now have all ingredients to state our problem and main results.
Let~$\varphi_A(X)$ and~$\varphi_B(X)$ be quantifier-free formulas with free variables $X = (X_1, \ldots, X_n)$.
Further, let $A := \{x \in \R^n \mid \varphi_A(x)\}$, $B := \{x \in \R^n \mid \varphi_B(x)\}$ and let $t \in \N$ be a natural number.
The \Hausdorff problem asks whether~$\dH(A,B) \leq t$.
Here the dimension~$n$ of the ambient space of~$A$ and~$B$ is part of the input.
We note that there is a polynomial-time algorithm for every fixed~$n$, see the related work in \cref{sec:related_work}.
Our main result determines the algorithmic complexity.

\begin{theorem}
    \label{thm:hausdorff_strict_uetr_complete}
    The \Hausdorff problem is \UER[<]-complete.
\end{theorem}

While we discuss the input encoding of~$\varphi_A$ and~$\varphi_B$ in \cref{sec:fotr}, let us briefly note here that~$t$ is encoded in binary.
A generalization to $t \in \Q$ is straight forward, by representing~$t$ as a reduced fraction and encoding the numerator and denominator in binary.
How to generalize further to $t \in \R$ is not clear, as the input must be a bitstring of finite length.

Prior to our result, it was not even known whether computing the Hausdorff distance is \NP-hard.
As \UER[<] contains the complexity classes \NP, \co\NP, \ER and \UR, our result implies hardness for these classes.
\Cref{thm:hausdorff_strict_uetr_complete} answers an open question posed by Dobbins, Kleist, Miltzow and Rz{\k{a}}{\.{z}}ewski~\cite{Dobbins2018_AreaUniversality}.

One may wonder whether it is crucial for our results that the  \Hausdorff problem asks if the distance $\leq t$ rather than $<t$. We remark that all our proofs work with tiny modifications also for the case of a strict inequality.
Furthermore, our results also hold for the directed Hausdorff distance.
Note that one can compute the undirected Hausdorff distance trivially, by computing twice the directed Hausdorff distance.
Thus intuitively, the directed Hausdorff distance is computationally at least as hard.
Yet, this is not a many-one reduction, as we need to compute the directed Hausdorff distance twice.

In the proof of \UER[<]-hardness for \cref{thm:hausdorff_strict_uetr_complete}, we create instances with some additional properties.
First, our reduction is a gap reduction and the Hausdorff distance of the obtained instance is either below the threshold~$t$ or at least $t \cdot \cramped{2^{2^{\Omega(n)}}}$, where~$n$ denotes the number of variables of~$\varphi_A$ and~$\varphi_B$.
Thus, our result also yields the following inapproximability result:

\begin{restatable}{corollary}{approx}
    \label{cor:inapproximability}
    Let~$A$ and~$B$ be two semi-algebraic sets in~$\R^n$ and~$f(n) = \cramped{2^{2^{o(n)}}}$.
    There is no polynomial-time $f(n)$-approximation algorithm to compute~$\dH(A,B)$, unless $\P = \UER[<]$.
\end{restatable}

Second, our reduction can be modified slightly to obtain a \Hausdorff instance in which~$A$ and~$B$ are described by relatively simple formulas.

\begin{restatable}{corollary}{simplicity}
    \label{cor:simple_formulas}
    The \Hausdorff problem remains \UER[<]-complete, even if the two sets~$A$ and~$B$ are both described either by
    \begin{enumerate}[(i)]
        \item\label{itm:simplicity_quadratic_equations} a conjunction of quadratic polynomial equations, or
        \item\label{itm:simplicity_single_equation} a single polynomial equation of degree at most four.
    \end{enumerate}
\end{restatable}
\medskip

Our last result concerns the complexity class \UER[<] itself.
As shown in \cref{thm:hausdorff_strict_uetr_complete}, the complexity class \UER[<] exactly captures the complexity of the \Hausdorff problem.
It is defined via the decision problem \StrictUETR which adds the syntactical restriction to \UETR of only allowing strict inequalities as atoms.
There are other complexity classes between \ER/\UR and \UER, one of them is described by Bürgisser and Cucker~\cite{Burgisser2009_ExoticQuantifiers}.
They define new quantifiers that make topological restrictions to \UETR.
Among others they define the \emph{exotic} quantifier~$\forall^*$ as
\[
    \forall^* X \in \R^n :
    \varphi(X)
    \wdequiv
    \forall X \in \R^n, \eps > 0 \sepdot
    \exists \til{X} \in \R^n :
    \norm{X - \til{X}}^2 < \eps^2 \land \varphi(\til{X})
    \text{.}
\]
Intuitively, $\forall^* X \in R^n : \varphi(X)$ does not require that~$\varphi(x)$ holds for all~$x\in \R^n$ but only for all $x$ in a dense subset of $\R^n$.
For more details, we refer to  \cref{sec:exotic_quantifiers}.
Let \ExoticUETR denote the decision problem whether a sentence of the form
\[
    \forall^* X \in \R^n \sepdot
    \exists Y \in \R^m :
    \varphi(X,Y)
\]
is true, where~$\varphi \in \QFF$.
Then the complexity class $\forall^*\exists\R$ contains all problems that many-one reduce to \ExoticUETR. 
Based on our techniques and developed tools, we can show that the two complexity classes $\forall^*\exists\R$ and $\UER[<]$ are the same.

\begin{restatable}{theorem}{exoticuetr}
    \label{thm:exotic_uetr}
    \ExoticUETR is \UER[<]-complete.
    Thus $\forall^*\exists\R = \UER[<]$.
\end{restatable}

\subsection{Related Work}
\label{sec:related_work}

This subsection reviews previous work concerning two directions. First, we discuss the complexity of computing the Hausdorff distance for special sets. Afterwards, we investigate previous work on the complexity class $\forall\exists\R$.

\subsubsection{Computing the Hausdorff Distance}
The notion of the Hausdorff distance was introduced  by Felix Hausdorff in 1914~\cite{Hausdorff1914_MengenLehre}.
Many early works focused on the Hausdorff distance for finite point sets.
For a set of~$a$ points and a set of~$b$ points 
in any fixed dimension, the Hausdorff distance can be computed by checking all pairs, i.e., in time~$O(ab)$.
In the plane, the runtime can be improved to $O((a+b)\log (a+b))$ by using Voronoi diagrams~\cite{Alt1995_PolygonalShapes}.
In fact, this method can be extended to sets consisting of pairwise non-crossing line segments in the plane, e.g., simple polygons and polygonal chains fulfill this property.
If the polygons are additionally convex, their Hausdorff distance can even be computed in linear time~\cite{Atallah1983_HausdorffConvexPolygons}.

More generally, the Hausdorff distance can be computed in polynomial time whenever the two sets can be described by a simplicial complex of fixed dimension.
Based on the PhD thesis of Godau~\cite{Godau1999_MeasuringSimilarity}, Alt et al.~\cite[Theorem 3.3]{Alt2003_Hausdorff} show how to compute the directed Hausdorff distance between two sets in~$\R^n$ consisting of~$a$ and~$b$ $k$-dimensional simplices in time $O(ab^{k+2})$ (assuming~$n$ is constant).
Using a Las Vegas algorithm for computing the vertices of the lower envelope, similar ideas yield an approach with randomized expected time in $O(ab^{k+\varepsilon})$ for $k > 1$ and every $\varepsilon > 0$~\cite[Theorem 3.4]{Alt2003_Hausdorff}.
They additionally present algorithms with better randomized expected running times for sets of triangles in~$\R^3$ and point sets in~$\R^n$.

Given two semi-algebraic sets $A, B \subseteq \R^n$ and a threshold value $t \in \N$, the \Hausdorff decision problem can be encoded as a \UETR sentence~$\Phi$ as already done for the directed Hausdorff distance in sentence \eqref{eqn:directed_hausdorff} above.
Such a sentence can be decided in time~$\cramped{(sd)^{O(n^2)}}$, where~$d$ denotes the maximum degree of any polynomial of~$\Phi$ and~$s$ denotes the number of atoms~\cite[Theorem 14.14]{Basu2006_RealAlgebraicGeometry}.

In other contexts the two sets are allowed to undergo certain transformations (e.g. translations) such that
the Hausdorff distance is minimized~\cite{Bringmann2021_TranslatingHausdorff}.
See Alt~\cite{Alt2000_DiscreteGeometricShapes} for a survey.

\subsubsection{The (Universal) Existential Theory of the Reals}

As mentioned above, the complexity class \UER was first studied by B\"{u}rgisser and Cucker who prove complexity results for many decision problems involving circuits~\cite{Burgisser2009_ExoticQuantifiers}.
Dobbins, Kleist, Miltzow, and Rz{\k{a}}{\.{z}}ewski~\cite{Dobbins2022_AreaUniversality,Dobbins2018_AreaUniversality} consider \UER in the context of area-universality of graphs.
A plane graph is \emph{area-universal} if for every assignment of non-negative reals to the inner faces of a plane graph, there exists a straight-line drawing such that the area of each inner face equals the assigned number.
Dobbins et al.\ conjecture that the decision problem whether a given plane graph is area-universal is complete for \UER.
They support this conjecture by proving hardness for several related notions~\cite{Dobbins2018_AreaUniversality}.
Additionally, for future research directions, they present a number of candidates for potentially \UER-hard problems.
Among them, they asked whether the \Hausdorff problem is \UER-complete.
The other candidates exhibit intrinsic connections to the notions of imprecision, robustness and extendability.

We point out that the computational complexity may also become easier when asking universal-type questions.
For example, it is \ER-complete to decide whether a graph is a unit distance graph, i.e., whether it has a straight-line drawing in the plane in which all edges have the same length~\cite{Schaefer2013_Realizability}.
On the other hand, the decision problem whether for all reasonable assignments of weights to the edges, a graph has straight-line drawing in which the edge lengths correspond to the assigned weight lies in \P~\cite{Belk2007_Realizability}.
Similarly, it is \ER-complete to decide for a given planar graph for which some vertices are fixed to the boundary of a polygon (with holes) whether there exists a planar straight-line drawing inside the polygon~\cite{Lubiw2022_DrawingInPolygonialRegion}.
The case of simple polygons is open.
In contrast, there is a polynomial time algorithm to test if a given graph~$G$ and a contained cycle~$C$ admit for \textit{every} simple polygon~$P$, representing~$C$, a straight-line drawing of~$G$ inside~$P$~\cite{Ophelders2021_PolygonUniversal}.

The complement class \EUR was recently investigated by D'Costa et al.~\cite{DCosta2021_EscapeProblem}.
They show that it is $\EUR[\leq]$-complete (where $\EUR[\leq] = \co\UER[<]$) to decide for a given rational matrix $A$ and a compact semi-algebraic set $K\subseteq \R^n$, whether there exists a starting point $x\in K$ such that $x_n:=A^nx$ is contained in $K$ for all $n\in N$.
This and similar problems are generally referred to as \emph{escape problems}.
Another subclass of \EUR, called $\exists^D\!\cdot\!\forall\R$ (here the~$\exists^D$ restricts the existentially quantified variables to Boolean instead of real values), was introduced by Blanc and Hansen~\cite{Blanc2021_ESS}.
They show that computing \emph{evolutionary stable strategies} in certain multi-player games is $\exists^D\!\cdot\!\forall\R$-complete.

We understand the complexity class \UER as a natural extension of
the complexity class~\ER (pronounced as
\enquote{exists R} or \enquote{ER}), which is defined similarly to \UER, but without universally quantified variables.
The complexity class \ER has gained a lot of interest in recent years, specifically in the computational geometry community.
It gains its significance because numerous well-studied problems from diverse areas of theoretical computer science and mathematics have been shown to be complete for this class.
Famous examples from discrete geometry are the recognition of geometric structures, such as unit disk graphs~\cite{McDiarmid2013_DiskSegmentGraphs}, segment intersection graphs~\cite{Matousek2014_IntersectionGraphsER}, visibility graphs~\cite{Cardinal2017_PointVisibilityGraphs}, stretchability of pseudoline arrangements~\cite{Mnev1988_UniversalityTheorem,Shor1991_Stretchability} and order type realizability~\cite{Matousek2014_IntersectionGraphsER}.
Other \ER-complete problems are related to graph drawing~\cite{Lubiw2022_DrawingInPolygonialRegion}, Nash-Equilibria~\cite{Bilo2016_Nash,Garg2018_MultiPlayer}, geometric packing~\cite{Abrahamsen2020_Framework}, the art gallery problem~\cite{Abrahamsen2022_ArtGallery}, convex covers~\cite{Abrahamsen2022_Covering}, non-negative matrix factorization~\cite{Shitov2016_MatrixFactorizations}, polytopes~\cite{Dobbins2019_NestedPolytopes,Richter1995_Polytopes}, geometric embeddings of simplicial complexes~\cite{Abrahamsem2021_GeometricEmbeddings}, geometric linkage constructions~\cite{Abel2016_PlaneRigidity}, training neural networks~\cite{Abrahamsen2021_NeuralNetworks,Bertschinger2022_NeuralNetworks}, and continuous constraint satisfaction problems~\cite{Miltzow2022_CCSP}.
We refer the reader to the lecture notes by Matou\v{s}ek~\cite{Matousek2014_IntersectionGraphsER} and surveys by Schaefer~\cite{Schaefer2010_GeometricTopoligical} and Cardinal~\cite{Cardinal2015_Survey} for more
information on the complexity class~\ER.

\subsection{Techniques and Proof Overview}
\label{sec:proof_overview}

In this subsection, we present the general idea behind the hardness reduction for the \Hausdorff problem.
The goal is to convey the intuition and to motivate the technical intermediate steps needed.
The sketched reduction  is oversimplified and thus neither in polynomial time nor fully correct. We point out both of these issues and give first ideas on how to solve them.

Let $\Phi := \forall X \in \R^n \sepdot \exists Y \in \R^m : \varphi(X,Y)$ be a \StrictUETR instance.
We define two sets
\begin{align*}
  A &:= \{x \in \R^n \mid \exists Y \in \R^m : \varphi(x,Y)\} \quad \text{and} \\
  B &:= \R^n
\end{align*}
and ask whether~$\dH(A,B) = 0$.
If~$\Phi$ is true, then~$A = \R^n$ and we have $\dH(A,B) = 0$ because both sets are equal.
Otherwise, if~$\Phi$ is false, then there exists some~$x \in \R^n$ for which there is no~$y \in \R^m$ satisfying~$\varphi(x,y)$ and we conclude that~$A \subsetneq \R^n$.
In general, we call the set of all~$x \in \R^n$ for which there is no~$y \in \R^m$ satisfying~$\varphi(x,y)$ the \emph{counterexamples}~$\counterexample{\Phi}$ of~$\Phi$.
One might hope that~$\counterexample{\Phi} \neq \emptyset$ is enough to obtain~$\dH(A,B) > 0$.
However, this is not the case.
To this end, consider the formula~$\Psi := \forall X \in \R \sepdot \exists Y \in \R : XY > 1$, which is false.
The set~$\counterexample{\Psi} = \{0\}$ contains only a single element, so we have~$A = \R \setminus \{0\}$ and $B = \R$.
However, their Hausdorff distance also evaluates to $\dH(A,B)=0$.
We conclude that above reduction does not (yet) work, because it
may also map no-instances of \StrictUETR to yes-instances of \Hausdorff.

We solve this issue by a preprocessing step that expands the set of counter\-examples. Specifically, \cref{thm:open_ball_counterexamples} establishes a polynomial-time algorithm to transform a \problemname{Strict-UETR} instance~$\Phi$ into an equivalent formula~$\Phi'$ such that the set of counterexamples is either empty (if $\Phi'$ is true) or contains an open ball of positive radius (if $\Phi'$ is false).
The radius of the ball serves as a lower bound on the Hausdorff distance~$\dH(A,B)$.
Thus a reduction starting with~$\Phi'$ is correct.
A key tool for this step is that we can restrict the variable ranges from $\R^n$ and $\R^m$ to small and compact intervals.
\cref{fig:counterexamples_motivation} presents an example on how such a range restriction may enlarge the set of counterexamples from a single point to an interval.

\begin{figure}[htb]
    \centering
    \begin{subfigure}[t]{0.47\textwidth}
        \centering
        \includegraphics[page=1]{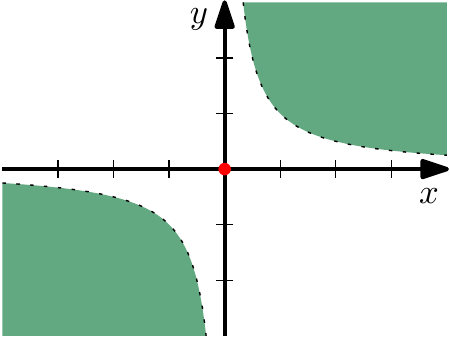}
        \caption{
            Each point~$(x,y) \in \R^2$ in the green open region satisfies~$xy > 1$.
            Only for $x = 0$ (in red) no suitable~$y \in \R$ exists.
        }
        \label{fig:counterexamples_motivation_point}
    \end{subfigure}
    \hfill
    \begin{subfigure}[t]{0.47\textwidth}
        \centering
        \includegraphics[page=2]{figures/counterexamples-motivation.pdf}
        \caption{If we restrict the range of~$Y$ to $[-1,1]$, then for no $x \in [-1,1]$ (in red) a suitable~$y$ with~$xy > 1$ exists.}
        \label{fig:counterexamples_motivation_interval}
    \end{subfigure}
    \caption{Expanding the set of counterexamples of~$\forall X \in \R \sepdot \exists Y \in \R : XY > 1$.}
    \label{fig:counterexamples_motivation}
\end{figure}

We emphasize that it is not known whether such a restriction of the variable ranges is possible for general \UETR formulas.
However, we exploit the fact that \StrictUETR formulas have a special property that we call \emph{\UStrict};
a negation- and implication-free formula is \UStrict if each atom involving universally quantified variables is a strict inequality.
Being \UStrict is a key property of many of the formulas considered throughout the paper and our proofs crucially rely on it.

A further challenge is given by the definition of the sets~$A$ and~$B$.
While the description complexity of~$B$ depends only on~$n$,
the definition of~$A$ contains an existential quantifier.
This is troublesome because our definition of the \Hausdorff problem requires quantifier-free formulas as its input, and in general there is no equivalent quantifier-free formula of polynomial length which describes the set~$A$~\cite{Davenport1988_QuantifierElimination}.
We overcome this issue by taking the existentially quantified variables as additional dimensions into account.
We scale them to a range much smaller than the range of the universally quantified variables, so that their influence on the Hausdorff distance becomes negligible.
Therefore instead of the above,  we work (in \cref{sec:hausdorff_hardness}) with sets similar to
\begin{align*}
    A &:= \{(x,y) \mid
        x \in [-C,C]^n, y \in [-1,1]^m, \varphi(x,y)
    \} \quad \text{and} \\
    B &:= [-C,C]^n \times \{0\}^m
\end{align*}
for some value~$C$ that is doubly exponentially large in~$\abs{\Phi}$.
This definition of~$A$ and~$B$ introduces the new issue that even if~$\Phi$ is true, the Hausdorff distance~$\dH(A,B)$ might be strictly positive.
However, we manage to identify a threshold~$t \in \N$, such that~$\dH(A,B) \leq t$ if and only if~$\Phi$ is true.
This completes the proof of \UER[<]-hardness.
In \cref{sec:erd_hausdorff} we also establish \UER[<]-hardness for~$t = 0$.

\subsection{Organization}
The remainder of the paper is organized as follows.
We introduce preliminaries concerning the first-order theory of the reals in \cref{sec:prelim} and essential tools from real algebraic geometry in \cref{sec:toolbox}.
\Cref{sec:counterexamples} presents the result for expanding the set of counterexamples for \UStrict formulas.
Finally, \cref{sec:hausdorff_hardness} contains the \UER[<]-hardness proof, followed by the \UER[<]-membership in \cref{sec:hausdorff_membership}.
In \cref{sec:exotic_quantifiers}, we apply our findings to so-called exotic quantifiers and relate them to \UER[<].
We conclude with a list of interesting open problems in \cref{sec:open_problems}.

\section{Preliminaries}
\label{sec:prelim}

In this section, we first introduce the necessary notation and definitions.
Afterwards, we consider the relation between all complexity classes that are relevant for the paper.

\subsection{First-Order Theory of the Reals}
\label{sec:fotr}

We mostly introduce standard terminology following the book by Cox, Little and O'Shea~\cite{Cox2006_AlgebraicGeometry}.
An \emph{atom} is an expression of the form~$P\,\circ\,0$ for some  polynomial $P \in \Z[X_1, \ldots, X_n]$ and~$\circ \in \{<, \leq, =, \neq, \geq, >\}$.
We always assume that a polynomial is written as a sum of monomials.
Its \emph{total degree} is the maximum number of occurrences of variables involved in any monomial.
For example $P(X,Y,Z) = X^2Y^2 + XYZ$ has total degree four.
A variable is called \emph{free} if it is not bound by a quantifier.
A \emph{formula} is either
\begin{itemize}
    \item an atom, or
    \item if $\varphi_1,\varphi_2$ are formulas, then their conjunction $\varphi_1 \land \varphi_2$, their disjunction $\varphi_1 \lor \varphi_2$, and the negation $\lnot \varphi_1$ are formulas, or
    \item  if~$\varphi(X)$ is a formula with a free variable~$X$, then $\exists X \in \R : \varphi(X)$ and $\forall X \in \R : \varphi(X)$ are formulas in which~$X$ is bound.
\end{itemize}
Note that the implication $\varphi_1 \implies \varphi_2$ can be formulated as $\lnot \varphi_1 \lor \varphi_2$.

In order to determine the \emph{length} $\abs{\varphi}$ of a formula $\varphi$, we count~$1$ for each fixed symbol, we encode integer coefficients in binary, exponents in unary, and we count $\log k$ (to encode the variable name) for each occurrence of a variable, where~$k$ denotes the number of variables.
We denote by \QFF the family of quantifier free formulas.
Furthermore, \QFF[<] and \QFF[\leq] are the families in \QFF that do not contain any negations and have only atoms involving~$<$ and~$\leq$, respectively.

A \emph{sentence} is a formula without free variables and thus either equivalent to true or to false.
As a convention, we use upper case Greek letters for sentences and use lower case Greek letter for formulas.
We write $\Psi \equiv \Psi'$ if the two sentences have the same truth value.
The \emph{first-order theory of the reals} is the family of all true sentences.
If all quantifiers of a formula appear at its beginning, we say it is in \emph{prenex normal form}.
We usually write \emph{blocks of variables}, i.e., $\forall X \in \R^n :  \varphi(X)$.
Here~$X$ is a shorthand notation for $X = (X_1, \ldots, X_n)$.
We say~$n$ is the length of~$X$ in this case.

We use upper case Latin letters for variables in formulas and lower case Latin letters for specific values, i.e., symbol~$X$ denotes a vector of variables, while~$x \in \R^n$ is a point.
We sometimes write~$\varphi(X,Y)$ to emphasize that~$X$ and~$Y$ are (blocks of) free variables in~$\varphi$.
Often we omit the free variables in~$\varphi$ though.

Consider a formula $\Phi \dequiv \forall X \in \R^n \sepdot \exists Y \in \R^m : \varphi(X,Y)$, where~$\varphi \in \QFF$.
Each atom of~$\varphi$ is of the form $P \circ 0$, where $\circ \in \{<, \leq, =, \neq, \geq, >\}$ and $P \in \Z[X,Y]$ is a polynomial in the variables~$X$ and~$Y$.
Without loss of generality we can restrict our attention to the case of $\circ \in \{<, \leq\}$, because the following transformations show that the other relations can be expressed by~$<$ or~$\leq$ such that the length of the formula is at most doubled:
\begin{align*}
    P > 0 & \wequiv -P < 0 &
    P = 0 & \wequiv (P \leq 0) \land (-P \leq 0) \\
    P \geq 0 & \wequiv -P \leq 0 &
    P \neq 0 & \wequiv  (P < 0) \lor (-P < 0)
\end{align*}
Furthermore, we can assume that~$\varphi$ contains only the logical connectives~$\land$ and~$\lor$, because De~Morgan's law allows to push all negations (and therefore also implications) down to the atoms transforming~$\varphi$ into \emph{negation normal form}.
This justifies, why \QFF[\circ] for $\circ \in \{<, \leq\}$ only contains negation-free formulas by our definition.
With the following equivalences we obtain a formula without negations:
\begin{align*}
    \neg(P < 0) & \wequiv -P \leq 0 &
    \neg(P \leq 0) & \wequiv -P < 0
\end{align*}

Given a formula~$\varphi(X)$ with~$n$ free variables~$X$, the set $S(\varphi) = \{x \in \R^n \mid \varphi(x)\}$ is semi-algebraic.
The \emph{complexity} of a semi-algebraic set~$S$ is the length of a shortest quantifier-free formula~$\varphi$, such that $S = S(\varphi)$ (recall that integer coefficients are encoded in binary).
We write $\varphi \equiv \varphi'$ if $S(\varphi) = S(\varphi')$.

\subsection{Complexity Classes}

Several complexity classes appear in this paper.
Here we discuss their relations among one another.
We make use of a helpful \lcnamecref{lem:tseitin} from the literature.
It allows us to replace a quantifier-free formula~$\varphi$ by a simpler one (at the cost of adding additional variables).
Here and throughout the rest of the paper the notation $x \leq \poly(y_1, \ldots, y_k)$ means that there is a polynomial $p \in \Z[Y_1, \ldots, Y_k]$ such that $x \leq p(y_1, \ldots, y_k)$.

\begin{lemma}[{\cite[Lemma~$3.2$]{Schaefer2017_FixedPointsNash}}]
    \label{lem:tseitin}
    Let~$\varphi(X) \in \QFF$ be a formula with~$n$ free variables~$X$.
    Then we can construct either of the following in polynomial time:
    \begin{enumerate}[(i)]
        \item\label{itm:tseitin_quadratic_equations}
        Integers $\ell, m \leq \poly(\abs{\varphi})$ and for $i \in \{1, \ldots, m\}$ a polynomial $F_i : \R^{n + \ell} \to \R$ of degree at most~$2$ such that
        \[
            \{x \in \R^n \mid \varphi(x)\} =
            \{x \in \R^n \mid
                \exists Y \in \R^\ell :
                \bigwedge_{i = 1}^m F_i(x,Y) = 0
            \}
            \text{.}
        \]
        \item\label{itm:tseitin_single_equation}
        An integer~$k \leq poly(\abs{\varphi})$ and a polynomial $F : \R^{n + k} \to \R$ of degree at most~$4$ such that
        \[
            \{x \in \R^n \mid \varphi(x)\} =
            \{x \in \R^n \mid \exists Y \in \R^k : F(x,Y) = 0\}
            \text{.}
        \]
    \end{enumerate}
\end{lemma}

For any fixed $\circ \in \{<, \leq\}$, we denote by~$\UER[\circ]$ the subset of~$\UER$ containing all decision problems that polynomial-time many-one reduce to a \UETR-instance whose quantifier-free parts are contained in \QFF[\circ].
Similarly, for $\circ \in \{<, \leq\}$, we denote the corresponding subsets of \ER and \UR  by~$\exists_{\circ}\R$ and $\forall_{\circ}\R$, respectively.
The following \lcnamecref{lem:UER_inclusions} summarizes what we know about the relation between the complexity classes \UER[<], \UER[\leq] and \UER as well as their relation to the well-studied classes \NP, \co\NP, \ER, \UR, and \PSPACE.

\begin{lemma}
    \label{lem:UER_inclusions}
    The following inclusions hold:
    \medskip
    \begin{center}
        \includegraphics{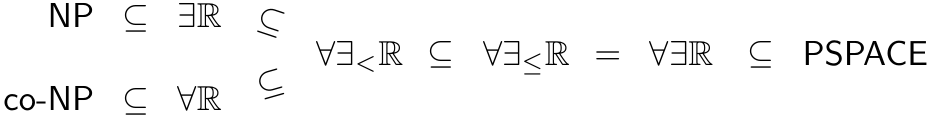}
    \end{center}
\end{lemma}

\begin{proof}
    The inclusion $\NP \subseteq \ER$ was first presented by Shor~\cite{Shor1991_Stretchability}.
    This directly implies $\co\NP \subseteq \UR$ (because $\UR = \co\ER)$.
    For $\circ \in \{<, \leq\}$ the inclusion $\UER[\circ] \subseteq \UER$ follows by definition because the left hand side is just a special case of the right hand side.
    Using that $\exists_<\R = \exists\R$~\cite[Theorem~$4.1$]{Schaefer2017_FixedPointsNash}, the same argument can be used for $\ER \subseteq \UER[<]$.
    Canny first established $\UER \subseteq \PSPACE$ in his seminal paper~\cite{Canny1988_PSPACE}.

    To show that $\UER \subseteq \UER[\leq]$, consider a \UETR instance
    \[
        \Phi \wdequiv
        \forall X \in \R^n \sepdot
        \exists Y \in \R^m :
        \varphi(X,Y)
        \text{.}
    \]
    We apply \cref{lem:tseitin} to~$\varphi$ and obtain in polynomial time an integer~$k \leq \poly(\abs{\varphi})$ and a polynomial $F : \R^{n+m+k} \to \R$, such that
    \[
        \Psi \wdequiv
        \forall X \in \R^n \sepdot
        \exists Y \in \R^{m+k} :
        F(X,Y) \leq 0 \land -F(X,Y) \leq 0
        \text{,}
    \]
    is equivalent to~$\Phi$.
    Note that both atoms use $\leq$.

    Lastly, let us consider the inclusion $\UR \subseteq \UER[<]$.
    Note that~$\UR = \forall_{<}\R$ (because two complexity classes are equal whenever their complement classes are equal and $\ER = \exists_{\leq}\R$ is known~\cite{Schaefer2017_FixedPointsNash}).
    Now~$\forall_{<}\R \subseteq \UER[<]$ again follows by definition.
\end{proof}

\section{Toolbox for Semi-Algebraic Sets}
\label{sec:toolbox}

In this section, we first introduce tools from real algebraic geometry, and then show how the ranges of quantifiers can be bounded.

\subsection{Tools Borrowed from Real Algebraic Geometry}

We review two sophisticated results from algebraic geometry, namely singly exponential quantifier elimination and the so called Ball Theorem.
While quantifier elimination provides equivalent quantifier-free formulas of bounded length, the Ball Theorem guarantees that every non-empty semi-algebraic set contains an element not too far from the origin.

We start with a result on quantifier-elimination which originates from a series of articles by Renegar~\cite{Renegar1992_QuantifierElimination1,Renegar1992_QuantifierElimination2,Renegar1992_QuantifierElimination3}.
Let us stress that the time complexity of this algorithm is singly exponential and not doubly exponential for every fixed number of quantifier alternations. 

\begin{theorem}[{\cite[Theorem~$14.16$]{Basu2006_RealAlgebraicGeometry}}]
    \label{thm:single_exponential_quantifier_elimination}
    Let~$X_1, \ldots, X_k, Y$ be blocks of real variables where~$X_i$ has length~$n_i$, $Y$ has length~$m$, formula $\varphi(X_1, \ldots, X_k, Y)\in \QFF$ has $s$ atoms and $Q_i\in\{\exists,\forall\}$ is a quantifier for all $i=1,\dots,k$.
    Further, let~$d$ be the maximum total degree of any polynomial of~$\varphi(X_1, \ldots, X_k, Y)$.
    Then for any formula
    \[
        \Phi(Y) \wdequiv
        Q_1 X_1 \in \R^{n_1} \ldots
        Q_k X_k \in \R^{n_k} :
        \varphi(X_1, \ldots, X_k, Y)
    \]
    there is an  equivalent quantifier-free formula of size at most
    \[
        s^{(n_1 + 1) \cdots (n_k + 1) (m + 1)}
        d^{O(n_1) \cdots  O(n_k) O(m)}
        \text{.}
    \]
\end{theorem}

Throughout the paper we use the following \lcnamecref{cor:single_exponential_quantifier_elimination} of \cref{thm:single_exponential_quantifier_elimination}.
It is weaker but easier to work with.

\begin{corollary}
    \label{cor:single_exponential_quantifier_elimination}
    Let~$\Phi(Y)$ be as in \cref{thm:single_exponential_quantifier_elimination} of length $L = \abs{\varphi(X_1, \ldots, X_k, Y)}$.
    Then for some constant~$\alpha \in \R$ independent of $\Phi$, there exists an equivalent quantifier-free formula of size at most
    \[
        L^{\alpha^{k+1} \cdot n_1\cdot \ldots \cdot n_k\cdot  m}
        \text{.}
    \]
\end{corollary}

\begin{proof}
  Let~$\Psi$ be the quantifier-free formula equivalent to~$\Phi$ obtained by \cref{thm:single_exponential_quantifier_elimination}.
  By observing that~$d,s \leq L$, we get
  \begin{align*}
    \abs{\Psi(Y)}
    &\leq
    L^{(n_1 + 1) \cdots (n_k + 1) (m + 1)} \cdot
    L^{O(n_1) \cdots O(n_k)  O(m)} \\
    &\leq
    L^{2n_1 \cdots 2n_k \cdot 2m} \cdot
    L^{O(n_1) \cdots O(n_k)  O(m)} \\
    &\leq
    L^{2n_1 \cdots 2n_k \cdot 2m + O(n_1) \cdots O(n_k) O(m)} \\
    &\leq
    L^{\alpha_1 n_1 \cdots \alpha_k n_k \cdot \alpha_m m}
    & & (\text{for $\alpha_1, \ldots, \alpha_k, \alpha_m \in \R$}) \\
    &\leq
    L^{\alpha^{k+1} n_1 \cdots n_k m}
    \text{,}
  \end{align*}
  where~$\alpha := \max\{\alpha_1, \ldots, \alpha_k, \alpha_m\}$.
\end{proof}

The Ball Theorem was first discovered by Vorob’ev~\cite{Vorobev1986_BallTheorem} and Grigor'ev and Vorobjov~\cite{Grigorev1988_BallTheorem}.
(Vorob’ev and Vorobjov are two different transcriptions of the same name from the Cyrillic to the Latin alphabet.)
Explicit bounds on the distance are given by Basu and Roy~\cite{Basu2010_BallTheorem}.
We use a formulation from Schaefer and \v{S}tefankovi\v{c}~\cite{Schaefer2017_FixedPointsNash}.

\begin{theorem}[Ball Theorem {\cite[Corollary~$3.1$]{Schaefer2017_FixedPointsNash}}]
  \label{thm:ball_theorem}
  Every non-empty semi-algebraic set in~$\R^n$ of complexity at most $L \geq 4$ contains a point of distance at most~$\cramped{2^{L^{8n}}}$ from the origin.
\end{theorem}

Recall that for any quantifier-free formula~$\varphi(X)$ with free variables~$X \in \R^n$, the set $S := \{x \in \R^n \mid \varphi(X)\}$ is semi-algebraic.
Thus, a direct conclusion of \cref{thm:ball_theorem} is that $\exists X \in \R^n : \varphi(X)$ is equivalent to $\exists X \in[\cramped{-2^{L^{8n}}}, \cramped{2^{L^{8n}}}]^n : \varphi(X)$. This is how we are going to use \cref{thm:ball_theorem} throughout this paper.

Below we use \cref{cor:single_exponential_quantifier_elimination} and \cref{thm:ball_theorem}, proving a \lcnamecref{lem:lower_bound_leading_eps} that was stated already by D'Costa, Lefaucheux, Neumann, Ouaknine and Worrel~\cite[Lemma~$14$]{DCosta2021_EscapeProblem} for two quantifiers.
We are interested in a generalization to more quantifiers.
The proof for the~$k$ quantifiers goes along the same lines as the proof for two quantifiers.

\begin{lemma}
    \label{lem:lower_bound_leading_eps}
    Let~$X_1, \ldots, X_k$ be blocks of variables where~$X_i$ has length~$n_i \geq 1$ and let $\varphi(\eps, X_1, \ldots, X_k) \in \QFF$ with~$L := \abs{\varphi}$.
    For~$Q_i \in \{\exists, \forall\}$ consider the set
    \[
        S := \{
            \eps > 0 \mid
            Q_1 X_1 \in \R^{n_1} \ldots Q_k X_k \in \R^{n_k} :
            \varphi(\eps, X_1, \ldots, X_k)
        \}
        \text{.}
    \]
    If~$S$ is non-empty, then there is an~$\eps^* \in S$ such that for some constant~$\beta \in \R$ we have
    \[
        \eps^* \geq 2^{-L^{\beta^{k+2} n_1 \cdots n_k}}
        \text{.}
    \]
\end{lemma}

\begin{proof}
    Let $\Phi(\eps)$ be the subformula $Q_1 X_1 \in \R^{n_1} \ldots Q_k X_k \in \R^{n_k} : \varphi(\eps, X_1, \ldots, X_k)$.
    By \cref{cor:single_exponential_quantifier_elimination}, there is a constant~$\alpha \in \R$ and a quantifier-free formula~$\phi(\eps)$ of length
    \[
        \abs{\phi(\eps)} \leq L^{2\alpha^{k+1} n_1 \cdots n_k}
    \]
    such that~$S = \{\eps > 0 \mid \phi(\eps)\}$.
    Let~$d$ be the maximum degree of any polynomial in~$\phi$ and~$\delta$ be a new variable.
    We replace each atom $P(\eps) \circ 0$ (where~$\circ \in \{<, \leq\}$) of~$\phi$ by $\delta^d P(1/\delta) \circ 0$ and denote the new formula by~$\psi(\delta)$.
    Then for~$\eps > 0$ it follows that~$\phi(\eps)$ is true if and only if for~$\delta = 1/\eps$ the sentence~$\psi(\delta)$ is true.
    It follows that
    \[
        \exists \eps > 0 : \phi(\eps)
        \wequiv
        \exists \delta > 0 : \psi(\delta)
        \text{.}
    \]
    To obtain a an upper bound on~$\abs{\psi(\delta)}$, note that the length of each atom increases by a factor of at most~$d$, which is obviously at most~$\abs{\phi(\eps)}$.
    We conclude that
    \[
        \abs{\psi(\delta)}
        \leq \abs{\phi(\eps)} \cdot d
        \leq \abs{\phi(\eps)}^2
        \text{.}
    \]
    If~$S$ is non-empty, then~$\exists \delta > 0 : \psi(\delta)$ is true.
    By \cref{thm:ball_theorem}, there is some~$\delta^*$ such that $\psi(\delta^*)$ is true and~$\delta^* \leq \cramped{2^{\abs{\psi(\delta)}^8}}$.
    We get that
    \[
        \delta^*
        \leq 2^{\abs{\psi(\delta)}^8}
        \leq 2^{\abs{\phi(\eps)}^{16}}
        \leq 2^{L^{32\alpha^{k+1} n_1 \cdots n_k}}
        \leq 2^{L^{\beta^{k+2} n_1 \cdots n_k}}
        \text{,}
    \]
    where~$\beta := \max\{32, \alpha\}$ is a real constant independent of the input.
    The result follows for~$\eps^* := 1/\delta^*$.
\end{proof}

The next \lcnamecref{lem:double_exponential_small_value} will be used frequently to scale (some dimensions of) semi-algebraic sets.
For some~$N \in \N$ and $N + 1$ variables $U = (U_0, \ldots, U_N)$ it considers the following formula:
\begin{equation}
    \label{eqn:double_exponetial_small_value}
    \chi(U) \wdequiv
    (2 \cdot U_0 = 1) \land
    \bigwedge_{i = 1}^N (U_i = U_{i-1}^2)
\end{equation}

\begin{lemma}
    \label{lem:double_exponential_small_value}
    For~$u \in [-1,1]^{N+1}$ formula~$\chi(u)$ is true if and only if~$u_i = \cramped{2^{-2^i}}$.
\end{lemma}

\begin{proof}
  The if-part is trivial.
  The only-if-part follows from a simple induction.
\end{proof}

\subsection{Bounding the Ranges of the Quantifiers}
\label{sec:range_restrictions}

In the following, we show how to restrict the ranges of the variables.
This was first done by D'Costa et al.~\cite{DCosta2021_EscapeProblem} in the context of their~$\EUR[\leq]$-complete escape problem.
\Cref{lem:bounding_universal_range,lem:bounding_existential_range} below are stated in our setting, but their proofs directly follow the ideas in~\cite{DCosta2021_EscapeProblem}.

As a first step, we restrict the universally quantified variables.
This works for general \UETR instances without any further requirements on the formula.

\begin{lemma}
    \label{lem:bounding_universal_range}
    Let~$X$ and~$Y$ be blocks of variables with $n := \abs{X}$ and~$m := \abs{Y}$, let $\varphi(X,Y) \in \QFF$ and let
    \[
        \Phi \wdequiv
        \forall X \in \R^n \sepdot
        \exists Y \in \R^m :
        \varphi(X,Y)
        \text{.}
    \]
    Then there exists an integer~$N \leq \poly(n, m, \abs{\varphi})$, such that for~$\cramped{C := 2^{2^N}}$ the sentence
    \[
        \Psi \wdequiv
        \forall X \in [-C,C]^n \sepdot
        \exists Y \in \R^m :
        \varphi(X,Y)
    \]
    is equivalent to~$\Phi$.
\end{lemma}

\begin{proof}
    We rewrite~$\Phi$ via a double negation to get
    \[
        \Phi \wequiv
        \neg\bigl(
            \exists X \in \R^n \sepdot
            \forall Y \in \R^m :
        \neg\varphi(X,Y)
        \bigr)
    \]
    and let~$L := \abs{\neg\varphi}$ denote the length of the quantifier-free part.
    By \cref{cor:single_exponential_quantifier_elimination} there is a constant $\alpha \in \R$ and a quantfier-free formula~$\psi(X)$ such that~$\Phi$ is equivalent to $\neg\bigl(\exists X \in \R^n : \psi(X)\bigr)$, where
    \[
        \abs{\psi}
        \leq L^{\alpha^2 n m}
        = 2^{\alpha^2 \log(L) n m}
        \text{.}
    \]
    Assuming that $\{x \in \R^n : \psi(x)\}$ is non-empty, \cref{thm:ball_theorem} yields that it contains a point of distance at most
    \[
        D := 
        2^{\abs{\psi}^{8n}}
        \leq 2^{(2^{\alpha^2 \log(L) n m})^{8n}}
        = 2^{2^{8 \alpha^2 \log(L) n^2 m}}
    \]
    from the origin.
    Let~$N = \bigl\lceil 8 \alpha^2 \log(L) n^2 m \bigr\rceil \leq \poly(n, m, \log(L)) \leq \poly(n, m, \abs{\varphi})$.
    Then it holds that~$C := \cramped{2^{2^N}} \geq D$.
    It follows, that
    \begin{align*}
        \neg\Phi
        &\wequiv \neg(
            \exists X \in \R^n :
            \psi(X)
        ) \\
        &\wequiv \neg(
            \exists X \in [-C,C]^n :
            \psi(X)
        ) \\
        &\wequiv \neg(
            \exists X \in [-C,C]^n \sepdot
            \forall Y \in \R^m :
            \varphi(X,Y)
        ) \\
        &\wequiv \neg\Psi
    \end{align*}
    and therefore~$\Phi \equiv \Psi$.
\end{proof}

In a second step, we additionally restrict the existentially quantified variables.
Before we do so, we show that this may be impossible in general (without changing its true/false value).
To this end, consider the following example: 
\[
    \forall X \in \R \sepdot
    \exists Y \in \R :
    X = 0 \lor XY = 1
\]
This sentence is clearly true as either $X = 0$ or if $X \neq 0$ we may define $Y := 1 / X$.
This remains true if we restrict the range of $X$, e.g., to $[-10,10]$.  However, note that $1/X$ with $X \in [-10,10]$ may be arbitrarily large (or small).
Consequently, we cannot restrict the range of~$Y$ to any interval.
In the following, we show how the ranges can be restricted in case of \UStrict formulas.
Requiring the formula to be \UStrict is a slight generalization of the corresponding statement shown in~\cite{DCosta2021_EscapeProblem} (where the formula is required to be strict).
This more general case is crucial for our proofs in \cref{sec:hausdorff_membership,sec:exotic_quantifiers}.

\begin{lemma}
    \label{lem:bounding_existential_range}
    Let~$X$ and~$Y$ be blocks of variables with~$n := \abs{X}$ and $m := \abs{Y}$ and $\varphi(X,Y) \in \QFF$.
    Further, let~$N$ be an integer and~$C := \cramped{2^{2^N}}$.
    Then for a \UStrict sentence
    \[
        \Phi \wdequiv
        \forall X \in [-C,C]^n \sepdot
        \exists Y \in \R^m :
        \varphi(X,Y)
    \]
    there is an integer~$M \leq \poly(n, m, N, \abs{\varphi})$ such that for~$D := \cramped{2^{2^{M}}}$ the sentence
    \[
        \Psi \wdequiv
        \forall X \in [-C,C]^n \sepdot
        \exists Y \in [-D,D]^m :
        \varphi(X,Y)
    \]
  is equivalent to~$\Phi$.
\end{lemma}

\begin{proof}
    If~$\Phi$ is false, then there exists an $x \in[-C,C]^n$ such that no $y \in \R^m$ satisfies $\varphi(x,y)$.
    In particular, no $y \in [-D,D]^m \subseteq \R^m$ satisfies~$\varphi(x,y)$.
    Thus, $\Psi$ is also false.

    In the remainder of the proof we assume that~$\Phi$ is true.
    The proof consists of two steps.
    First we show that an upper bound~$D$ for the existentially quantified variables indeed exists.
    In a second step, we use the Ball Theorem to compute an upper bound for~$D$.
    
    For the first step, let~$S = [-C,C]^n$.
    Sentence~$\Phi$ being true implies that for each~$x \in S$ there is a~$y(x) \in \R^m$ such that~$\varphi(x,y(x))$ is true.
    Even stronger, as~$\varphi$ is \UStrict, we even find an $\eps(x) > 0$, such that for all~$\til{x} \in S$ with $\norm{x - \til{x}} < \eps(x)$ we get that~$\varphi(\til{x},y(x))$ is true.
    Recall that we denote by $B_n(x,r) = \{\til{x} \in \R^n \mid \norm{\til{x} - x} < r\}$ the open ball with center~$x$ and radius~$r$ in~$\R^n$.
    Then $\{B_n(x,\eps(x)) \mid x \in S\}$ is an open cover of~$S$.
    As~$S$ is compact, it has a finite subcover $B_n(x_1,\eps(x_1)), \ldots, B_n(x_s,\eps(x_s))$.
    Now, given some $x \in S$, there is an $i \in \{1, \ldots, s\}$, such that $\varphi(x, y(x_i))$ is true.
    We define $y_{\max} := \max\{\norm{y(x_1)}_\infty, \ldots, \norm{y(x_s)}_\infty\}$.
    Then, for all~$D \geq y_{\max}$ formula~$\Phi$ implies
    \[
        \exists D > 0 \sepdot
        \forall X \in [-C,C]^n \sepdot
        \exists Y \in \R^m :
        \bigwedge_{i=1}^m \abs{Y_i} \leq D \land
        \varphi(X,Y)
        \text{,}
    \]
    proving the existence of an upper bound~$D$ for the existentially quantified variables.
    
    The second step is to obtain a bound on~$D$.
    We first need to construct~$C = \cramped{2^{2^N}}$ inside the formula.
    For this, let~$U = (U_0, \ldots, U_N)$ be~$N + 1$ new variables and~$\chi(U)$ be the formula \eqref{eqn:double_exponetial_small_value}. Recall that by \cref{lem:double_exponential_small_value}, $\chi(u)$ is true if and only if~$u_i = \cramped{2^{-2^i}}$.
    Further, each~$U_i$ can be trivially restricted to be in~$[-1,1]$.
    Using~$\chi(U)$, we can rewrite above sentence as
    \begin{align*}
        \exists D > 0 \sepdot
        &\forall X \in \R^n, U \in [-1,1]^{N+1} \sepdot
        \exists Y \in \R^m : \\
        &\bigl(
            \chi(U) \land
            \bigwedge_{i=1}^n \abs{X_i} U_N \leq 1
        \bigr) \implies
        \bigwedge_{i=1}^m \abs{Y_i} \leq D \land
        \varphi(X,Y)
        \text{.}
    \end{align*}
    
    From here on, bounding~$D$ is a straightforward application of the Ball Theorem:
    Let~$L$ be the length of the subformula behind the existential quantification of~$D$.
    By \cref{cor:single_exponential_quantifier_elimination} there is a constant~$\alpha \in \R$ and a quantifier-free formula~$\psi(D)$, such that above sentence is equivalent to $\exists D > 0 : \psi(D)$ where $\abs{\psi} \leq \cramped{L^{\alpha^3 (n + N + 1) m}}$.
    Then \cref{thm:ball_theorem} yields the following upper bound for~$D$:
    \[
        D
        \leq 2^{\abs{\psi(D)}}
        \leq 2^{(L^{\alpha^3 (n+N+1) m})^{8n}}
        = 2^{2^{8 \alpha^3 \log(L) n (n+N+1) m}}
        \text{.}
    \]
    Lastly, we choose~$M = \bigl\lceil 8 \alpha^3 \log(L) n (n+N+1) m \bigr\rceil$ to be the smallest integer such that~$D \leq \cramped{2^{2^M}}$.
    Note that $M \leq \poly(n, m, N, \log L) \leq \poly(n, m, N, \abs{\varphi})$ as required.
\end{proof}

\section{Counterexamples of \problemname{Strict-UETR}}
\label{sec:counterexamples}

Let us recall the definition of \emph{counterexamples} here that was already motivated in \cref{sec:proof_overview}.
Given a sentence~$\Phi \dequiv \forall X \in \R^n \sepdot \exists Y \in \R^m : \varphi(X,Y)$ we call
\[
    \counterexample{\Phi} := \{
        x \in \R^n \mid
        \forall Y \in \R^m :
        \neg\varphi(x,Y)
    \}
\]
its \emph{counterexamples}.
The counterexamples of~$\Phi$ are exactly the values~$x \in \R^n$ for which there is no~$y \in \R^m$ such that~$\varphi(x,y)$ is true.
The main result of this section, \cref{thm:open_ball_counterexamples}, is that we can transform a \problemname{Strict-UETR} instance~$\Phi$ into an equivalent formula~$\Psi$ for which~$\counterexample{\Psi}$ is either empty or contains an open ball.
The main tools for this are the range restrictions from \cref{sec:range_restrictions} and the following \lcnamecref{lem:continuity_of_compact_minimum} from calculus.

\begin{lemma}
    \label{lem:continuity_of_compact_minimum}
    Let~$S$ and~$T$ be compact sets and $f : S \times T \to \R$ be a continuous function.
    Then $g : S \to \R$, $x \mapsto \min\limits_{y \in T} \{f(x,y)\}$ is continuous over~$S$.
\end{lemma}

\begin{proof}
    We first observe that by compactness of~$S$ and~$T$, their Cartesian product~$S \times T$ is compact as well.
    Thus, because~$f$ is continuous on~$S \times T$, it is even uniformly continuous, i.e., for every~$\eps > 0$ there is a~$\delta > 0$, such that $\lvert f(x,y) - f(\til{x},\til{y}) \rvert < \eps$ whenever $\norm{(x,y) - (\til{x},\til{y})} < \delta$ for every two points~$(x,y),(\til{x},\til{y}) \in S \times T$.

    Now consider $x, \til{x} \in S$ with $\norm{x - \til{x}} < \delta$.
    We have
    \begin{align*}
        g(\til{x}) - g(x)
        &= g(\til{x}) - f(x, y)  & & \text{(for some $y \in T$)} \\
        &< g(\til{x}) - (f(\til{x}, y) - \eps) & & \text{(by uniform continuity)} \\
        &\leq g(\til{x}) - (g(\til{x}) - \eps) & & \text{(by definition of $g$)} \\
        &= \eps
        \text{.}
    \end{align*}
    By exchanging the role of~$x$ and~$\til{x}$, we get $g(x) - g(\til{x})< \eps$.
    Combined, we obtain that~$\lvert g(x) - g(\til{x}) \rvert < \eps$ for all~$x,\til{x} \in S$ with $\norm{x - \til{x}} < \delta$.
    It follows that~$g$ is continuous on~$S$.
\end{proof}

With these tools at hand, we are able to tackle the main result of this section.

\begin{theorem}
    \label{thm:open_ball_counterexamples}
    Given a \StrictUETR instance~$\Phi$, we can construct in polynomial time an equivalent \UETR instance~$\Psi$ of the form
    \[
        \forall X \in [-1,1]^n \sepdot
        \exists Y \in [-1,1]^\ell :
        \psi(X,Y)
        \text{,}
    \]
    such that~$\counterexample{\Psi}$ is either empty or contains an $n$-dimensional open ball.
\end{theorem}

\begin{proof}
    The proof is split into two parts.
    First, we construct~$\Psi$ from~$\Phi$.
    Afterwards, we show that~$\counterexample{\Psi}$ has the desired properties.

    \subparagraph*{Construction of~$\Psi$.}
    Without loss of generality, each atom of the sentence
    \[
        \Phi \wdequiv
        \forall X \in \R^n \sepdot
        \exists Y \in \R^m :
        \varphi_<(X,Y)
    \]
    with~$\varphi_< \in \QFF[<]$ is of the form~$P < 0$, where $P \in \Z[X,Y]$ is a polynomial.
    We replace~$P < 0$ by the equivalent formula $\exists Z \in \R : Z^2P + 1 < 0$.
    Here,~$Z$ is a new variable that is exclusive to this atom.
    While this replacement may look innocent, it is very powerful.
    The key insight is as follows:
    Once we bound the range of~$Z$ to some compact interval~$[-D,D]$, this requires~$P < -1/D^2$ in order to satisfy the atom.
    This is stronger than just requiring~$P < 0$ and expands the set of counterexamples as we show below.
    Moving the new existential quantifiers to the front yields a sentence
    \[
        \Phi_1 \wdequiv
        \forall X \in \R^n \sepdot
        \exists Y \in \R^m, Z \in \R^k :
        \varphi_<'(X,Y,Z)
    \]
    in prenex normal form.
    Here, $k$ is the number of atoms in~$\Phi$ and~$\varphi_<'$ is obtained from~$\varphi_<$ by above transformation.
    Note that in particular, $\varphi_<$ and~$\varphi_<'$ have exactly the same logical structure (their only difference lies in the transformed atoms).
    The length increases only by a constant amount per atom, so~$\abs{\Phi_1}$ is linear in~$\abs{\Phi}$.
    Further, we have~$\counterexample{\Phi_1} = \counterexample{\Phi}$ by construction.
    
    We can apply \cref{lem:bounding_universal_range} to restrict the ranges of the universally quantified variables and obtain an integer~$N \leq \poly(\abs{\Phi_1})$ such that for~$C := \cramped{2^{2^N}}$ the sentence~$\Phi_1$ is equivalent to
    \[
        \Phi_2 \wdequiv
        \forall X \in [-C,C]^n \sepdot
        \exists Y \in \R^m, Z \in \R^k :
        \varphi_<'(X,Y,Z)
        \text{.}
    \]
    It holds that $\counterexample{\Phi_2} \subseteq \counterexample{\Phi_1}$ and further that $\counterexample{\Phi_2} = \counterexample{\Phi_1} \cap [-C,C]^n$.
    
    Each atom in~$\Phi_2$ is a strict inequality.
    Thus, we can use \cref{lem:bounding_existential_range} to also restrict the ranges of the existentially quantified variables.
    We obtain another integer $M \leq \poly(N, \abs{\Phi_2})$ such that for $D := \cramped{2^{2^M}}$ above sentence~$\Phi_2$ is equivalent to
    \[
        \Phi_3 \wdequiv
        \forall X \in [-C,C]^n \sepdot
        \exists Y \in [-D,D]^m, Z \in [-D,D]^k :
        \varphi_<'(X,Y,Z)
        \text{.}
    \]
    Regarding the counterexamples, we have~$\counterexample{\Phi_3} \supseteq \counterexample{\Phi_2}$.
    
    The last step is to scale the ranges over which the variables are quantified to the interval~$[-1,1]$.
    To this end, define $K := \max\{N,M\}$, let $U := (U_0, \ldots, U_K)$ be~$K+1$ new variables and let~$\chi(U)$ be formula~\eqref{eqn:double_exponetial_small_value}.
    Recall that by \cref{lem:double_exponential_small_value}, for~$u \in [-1,1]^{K+1}$ we have~$\chi(u)$ if and only if~$u_i = \cramped{2^{-2^i}}$.
    Let~$d$ be the maximum degree of any polynomial in~$\varphi_<'$.
    We define
    \begin{align*}
        \Psi \wdequiv~
        &\forall X \in [-1,1]^n \sepdot
        \exists Y \in [-1,1]^m, Z \in [-1,1]^k, U \in [-1,1]^K : \\
        &\chi(U)~\land~U_K^d \cdot \varphi_<'\Bigl(
            \frac{X}{U_N}, \frac{Y}{U_M}, \frac{Z}{U_M}
        \Bigr)
        \text{,}
    \end{align*}
    where~$X/U_N$ means that every~$X_i$ is replaced by $X_i/U_N$ (likewise for~$Y/U_M$ and~$Z/U_M$).
    The multiplication of~$\varphi_<'$ with~$U_K^d$ denotes that both sides of each atom are multiplied by~$U_K^d$.
    This restores the requirement that each atom is a polynomial inequality.
    (Strictly speaking, the obtained formula contains the divisions by~$U_N$ and~$U_M$.
    However, because~$K \geq N$, we can replace any $U_K \cdot (X_i/U_N)$ by $U_{K-N}X_i$, which does not contain divisions.
    Likewise, we handle $Y/U_M$ and $Z/U_M$.)
    As this last step just scaled variables, we conclude that~$\Psi$ is equivalent to~$\Phi_3$ and therefore also to~$\Phi$.
    Also, sentence~$\Psi$ has the form required by the statement of the \lcnamecref{thm:open_ball_counterexamples}.
    
    \subparagraph*{Properties of~$\counterexample{\Psi}$.}
    It remains to show that~$\counterexample{\Psi}$ is either empty (if~$\Psi$ is true) or contains an $n$-dimensional open ball (if~$\Psi$ is false).
    Note that scaling variables (as done to get from~$\Phi_3$ to~$\Psi$) also scales the counterexamples;
    thus, an open ball in~$\counterexample{\Phi_3}$ maps to an open ball in~$\counterexample{\Psi}$.
    It therefore suffices to prove that~$\counterexample{\Phi_3}$ contains an open ball.
    As~$\Phi_3$ is the simpler formula, we analyze~$\counterexample{\Phi_3}$ below.
    
    By construction, $\Phi$ and~$\Phi_3$ are equivalent.
    Thus,~$\Phi$ is true if and only~$\Phi_3$ is true.
    In particular, $\counterexample{\Phi} = \emptyset$ implies that $\counterexample{\Phi_3} = \emptyset$.
    From now on, we assume that~$\Phi_3$ is false.
    Let~$x^* \in \counterexample{\Phi_2}$ be a counterexample of~$\Phi_2$, fixed until the end of the proof.
    We know that~$x^* \in [-C,C]^n$ (by \cref{lem:bounding_universal_range}) and also that~$x^* \in \counterexample{\Phi_3}$ (by construction of~$\Phi_3$).
    
    We proof below that for some~$r > 0$, all $x \in [-C,C]^n$ with $\norm{x^* - x} < r$ are counterexamples of~$\Phi_3$ as well.
    If $B_n(x^*, r) \subseteq [-C,C]^n$, then~$x^*$ is the center of our desired open ball of counterexamples.
    If~$B_n(x^*, r)$ is not completely contained in~$[-C,C]^n$, then any~$x' \in B_n(x^*, r) \cap (-C,C)^n$ can be used instead as the center of a smaller (but still open) ball of counterexamples.
    
    To simplify the following argument, we further assume that~$\varphi_<$ (in~$\Phi$) is in disjunctive normal form (DNF), i.e., a disjunction of conjunctions of atoms.
    By construction, $\varphi_<'$ is then also in DNF and has exactly the same logical structure.
    This is justified as the set of counterexamples is invariant under applications of the distributive law on the quantifier-free part.
    Thus, $\varphi_<$ and~$\varphi_<'$ have exactly the same counterexamples as their DNFs.
    
    Let $\mathcal{C}(X,Y) := \bigl(\bigwedge_{i=1}^s P_i(X,Y) < 0\bigr)$ be one of the conjunctive clauses of (the DNF of)~$\varphi_<(X,Y)$.
    For our fixed counterexample~$x^* \in \counterexample{\Phi_2}$, every conjunctive clause of $\varphi_<(x^*, Y)$ evaluates to false independently of~$Y$.
    We get that for all~$y \in \R^m$ and thus in particular for all~$y \in [-D,D]^m$ that~$\mathcal{C}(x^*,y)$ is false and that
    \begin{equation}
        \label{eqn:disjunction_non_negative}
        \bigvee_{i=1}^s \bigl(P_i(x^*,y) \geq 0\bigr)
    \end{equation}
    is true.
    Let us point out that for different choices of~$y \in [-D,D]^m$, different subsets of the polynomials~$P_i(x^*,y)$ may evaluate to non-negative values.
    We only know that for every~$y$ at least one of the polynomials is non-negative (here it is important that~$\varphi_<(X,Y)$ is in DNF).
    To overcome this we combine the polynomials into a single function.
    
    Each of the~$P_i \in \Z[X,Y]$, $i \in \{1, \ldots, s\}$, is a polynomial and thus continuous.
    The maximum over a finite number of continuous functions is again continuous, so
    \begin{align*}
        P_{\max} : [-C,C]^n \times [-D,D]^m &\to \R \\
        (x,y) &\mapsto \max\limits_{i=1,\ldots,s} \{P_i(x,y)\}
    \end{align*}
    is continuous.
    It follows from~\eqref{eqn:disjunction_non_negative} that for our fixed counterexample~$x^*$ and all $y \in [-D,D]^m$ it holds that
    \begin{equation}
        \label{eqn:maximum_non_negative}
        P_{\max}(x^*,y) \geq 0
        \text{.}
    \end{equation}
    We want to argue about the value of~$P_{\max}$ at points~$x$ in a small neighborhood around~$x^*$.
    To this end, we consider the function
    \begin{align*}
        P^* : [-C,C]^n &\to \R \\
        x &\mapsto \min\limits_{y \in [-D,D]^m} \  P_{\max}(x,y)
        \text{,}
    \end{align*}
    which eliminates the dependency on~$y$.
    The sets~$[-C,C]^n$ and~$[-D,D]^m$ are compact, so by \cref{lem:continuity_of_compact_minimum} function~$P^*$ is again continuous.
    From~\eqref{eqn:maximum_non_negative}, we get for our fixed counterexample~$x^*$ that
    \[
        P^*(x^*) \geq 0
        \text{.}
    \]
    By the continuity of~$P^*$, for every~$\eps > 0$ there exists a~$\delta > 0$ such that for all $x \in [-C,C]^n$ with~$\norm{x^* - x} < \delta$ we have $\abs{P^*(x) - P^*(x^*)} < \eps$.
    We choose~$\eps < 1/D^2$ and conclude that for a sufficiently small~$\delta > 0$ and all~$x \in [-C,C]^n$ with~$\norm{x^* - x} < \delta$ it holds that
    \[
        P^*(x) > -\frac{1}{D^2}
        \text{.}
    \]
    Fix one such~$x$.
    Going backwards through our chain of defined functions, it follows for all~$y \in [-D,D]^m$ that $P_{\max}(x,y) > -1/D^2$ and moreover that
    \begin{equation}
        \label{eqn:lower_bound_on_polynomials}
        \bigvee_{i = 1}^s P_i(x, y) > -\frac{1}{D^2}
        \text{.}
    \end{equation}
    Now also fix an arbitrary~$y \in [-D,D]^m$ and choose~$j \in \{1, \ldots, s\}$ such that $P_j(x,y) > -1/D^2$.
    Because $A := (P_j(X,Y) < 0)$ is an atom in the DNF of~$\varphi_<(X,Y)$, there is a corresponding atom $A' := (Z_j^2P_j(X,Y) + 1 < 0)$ in the DNF of~$\varphi_<'(X,Y,Z)$.
    Recall that~$Z_j$ is an existentially quantified variable that only appears in~$A'$.
    Note that~$A'$ can never be true for~$Z_j = 0$.
    For~$Z_j \neq 0$, the atom~$A'$ can be rewritten as $P_j(X,Y) < -1/Z_j^2$.
    From~$Z_j \in [-D,D]$, we get that~$Z_j^2 \leq D^2$ and therefore our considered atom~$A'$ can only ever be satisfied, if~$P_j(X,Y) < -1/Z_j^2 \leq -1/D^2$.
    However, by the choice of~$j$ and \eqref{eqn:lower_bound_on_polynomials}, we know that $P_j(x, y) > -1/D^2$.
    Thus, because~$y$ was fixed arbitrarily,~$x$ must be a counterexample of~$\Phi_3$.
    Additionally, because~$x \in \R^n$ with~$\norm{x^* - x} < \delta$ was arbitrary, we conclude that all such~$x$ are counterexamples of~$\Phi_3$ forming an~$n$-dimensional open ball.
\end{proof}

\section{\texorpdfstring{\UER[<]}{UE<R}-Hardness of \Hausdorff}
\label{sec:hausdorff_hardness}

We are now able to show \UER[<]-hardness.

\begin{theorem}
    \label{thm:hausdorff_strict_uetr_hard}
    \Hausdorff and \DirectedHausdorff are \UER[<]-hard.
\end{theorem}

\begin{proof}
    Let~$\Phi$ be an instance of \StrictUETR.
    We give a polynomial-time many-one reduction to an equivalent \Hausdorff instance.
    The proof is split into two parts:
    In the first part, we transform~$\Phi$ into an equivalent \UETR instance~$\Psi$ whose counterexamples~$\counterexample{\Psi}$ contain an open ball (if there are any).
    Sentence~$\Psi$ is then used to construct a \Hausdorff instance~$(A,B,t)$.
    The second part then proves that~$\Phi$ and~$(A,B,t)$ are indeed equivalent.
    
    \subparagraph*{Constructing \Hausdorff instance~$(A,B,t).$}
    The first step is to apply \cref{thm:open_ball_counterexamples} to~$\Phi$ and to obtain in polynomial time an equivalent \UETR instance
    \[
        \Psi' \wdequiv
        \forall X \in [-1,1]^n \sepdot
        \exists Y \in [-1,1]^m :
        \psi'(X,Y)
        \text{,}
    \]
    where~$\psi' \in \QFF$.
    We know that either~$\counterexample{\Psi'} = \emptyset$ (if~$\Psi'$ is true) or that $\counterexample{\Psi'}$ contains an $n$-dimensional open ball (if~$\Psi'$ is false).
    Based on~$\Psi'$ we define
    \begin{align*}
        \psi(X,Y) &\wdequiv \psi'(X,Y) \lor \bigwedge_{i=1}^n X_i = 0
        \quad \text{and} \\
        \Psi &\wdequiv
        \forall X \in [-1,1]^n \sepdot
        \exists Y \in [-1,1]^m :
        \psi(X,Y)
        \text{.}
    \end{align*}
    Note that~$\Psi'$ and~$\Psi$ are equivalent:
    If~$\Psi'$ is true, then obviously~$\Psi$ is also true because the new condition is added using a logical \enquote{or}.
    If~$\Psi'$ is false, then $\counterexample{\Psi} = \counterexample{\Psi'} \setminus \{\vec{0}\}$.
    Since~$\counterexample{\Psi'}$ contained an open ball, it follows that $\counterexample{\Psi}$ also contains an open ball.
    The key idea behind the definition of~$\Psi$ is that~$\counterexample{\Psi}$ is guaranteed to be a strict subset of~$\R^n$.
    This will be important below to make sure that set~$A$ (of the \Hausdorff instance we define below) is non-empty.
    
    If~$\Psi$ is false, then there is an~$x \in \counterexample{\Psi} \subseteq [-1,1]^n$, such that $B_n(x, r) \subseteq \counterexample{\Psi}$ for some~$r > 0$.
    Expressed as a sentence in the first-order theory of the reals we get
    \[
        \exists r > 0, X \in [-1,1]^n \sepdot
        \forall \til{X} \in [-1,1]^n, Y \in [-1,1]^m :
        \norm{X - \til{X}}^2 < r^2 \implies \neg\psi(\til{X}, Y)
        \text{.}
    \]
    Let us denote by~$L$ the length of the quantifier-free part of this sentence.
    We see that~$L$ is clearly polynomial in~$\abs{\Psi}$ which by construction is polynomial in~$\abs{\Phi}$.
    Above sentence has the form required by \cref{lem:lower_bound_leading_eps}, and we get a constant~$\beta \in \R$ such that the following lower bound for~$r$ can be assumed:
    \begin{equation}
        \label{eqn:lower_bound_radius}
        r \geq 2^{-L^{\beta^4 n (n+m)}}
    \end{equation}
    Let~$N = \bigl\lceil \beta^4 n (n+m) \bigr\rceil$ be the smallest integer, such that
    \begin{align}
        \label{eqn:definition_of_N}
        r \cdot 2^{2^N} > m
        \text{.}
    \end{align}
    By~\cref{eqn:lower_bound_radius}, it holds that $N \leq \poly(n, m, \log(L)) \leq \poly(\abs{\Phi})$.
    Define~$C := \cramped{2^{2^N}}$.
    
    The idea now is to scale the universally quantified variables by a factor of~$C$ (so that they are from the interval~$[-C,C]$).
    This then also scales the set of counterexamples~$\counterexample{\Psi}$ by~$C$ and in particular the radius of the open ball in~$\counterexample{\Psi}$.
    Let $U = (U_0, \ldots, U_N) \in [-1,1]$ be $N+1$ new variables and~$\chi(U)$ be formula~\eqref{eqn:double_exponetial_small_value}.
    Recall that by \cref{lem:double_exponential_small_value}, for~$u \in [-1,1]^{N+1}$ we have~$\chi(u)$ if and only if $u_i = \cramped{2^{-2^i}}$.
    With this, we define
    \[
        \phi(X,Y,U) \wdequiv \chi(U) \land \psi(U_N X, Y)
        \text{,}
    \]
    where $U_N X$ means that every occurrence of~$X_i$ in~$\psi$ is replaced by~$U_N X_i$.
    Finally we are ready to define our desired \Hausdorff instance:
    \begin{align*}
        A\ &:=\ \{
            (x,y,u) \in [-C,C]^n \times
            [-1,1]^{m} \times
            \{2^{-2^0}\} \times \ldots \times \{2^{-2^N}\} \mid
            \phi(x,y,u)
        \} \\
        B\ &:=\
            [-C,C]^n \times
            \{0\}^m \times
            \{2^{-2^0}\} \times \ldots \times \{2^{-2^N}\} \\
        t\ &:=\ m
    \end{align*}
    Note that this is well-defined, because both sets~$A$ and~$B$ are non-empty.
    While this is trivial for~$B$, it holds for~$A$ by our construction of~$\phi$ from~$\Phi$:
    It always holds that
    \[
        \emptyset
        \ \neq\
        \{0\}^{n} \times [-1,1]^m \times \{2^{-2^0}\} \times \ldots \times \{2^{-2^N}\}
        \ \subseteq\ 
        A
        \text{.}
    \]
    
    \subparagraph*{Equivalence of~$\Phi$ and~$(A,B,t).$}
    To see that~$\Phi$ and~$(A,B,t)$ are equivalent, assume first that~$\Phi$ is true.
    For every point~$a := (x, y, u) \in A$ it must hold that~$u_i = \cramped{2^{-2^i}}$ as this is necessary to satisfy~$\chi(u)$.
    Consider the point $b := (x, \{0\}^n, u) \in B$.
    We get
    \[
        \norm{a - b}
        = \norm{(x,y,u) - (x,\{0\}^m,u)}
        = \norm{y - \vec{0}}
        \leq \textstyle{\sqrt{\sum_{i=1}^m 1}}
        = \sqrt{m}
        \leq m
        = t
        \text{.}
    \]
    As~$a$ was chosen arbitrarily, we get an upper bound for the directed Hausdorff distance $\dDH(A,B) \leq t$.
    On the other hand, consider an arbitrary point~$b := (x, \{0\}^m, u) \in B$.
    Because~$\Phi$ (and therefore~$\Psi$) is true, there is some $y \in [-1,1]^m$ such that there is a point $a := (x, y, u) \in A$.
    By the same calculation as above, we get $\dDH(B,A) \leq t$ and thus
    \begin{equation}
        \label{eqn:hausdorff_upper_bound}
        \dH(A,B) \leq t
        \text{.}
    \end{equation}
    
    Now assume that~$\Phi$ and~$\Psi$ are false.
    Then there is some~$x \in [-1,1]^n$ such that there is an $n$-dimensional open ball~$B_n(x, r) \subseteq \counterexample{\Psi}$ (the~$r$ here is the one from \eqref{eqn:lower_bound_radius}).
    By the construction of~$A$, this corresponds to an open ball of radius~$C \cdot r$ in~$\R^n \setminus A$.
    Let~$x^*$ be the center of this open ball in~$\R^n \setminus A$.
    Then for~$b := (x^*, \{0\}^m, u) \in B$ all points~$a \in A$ have
    \[
        \norm{a - b}
        \geq C \cdot r
        > m
        = t
        \text{.}
    \]
    It follows that
    \begin{equation}
        \label{eqn:hausdorff_lower_bound}
        \dH(A,B)
        \geq \dDH(B,A)
        \geq \norm{a - b}
        > t.
    \end{equation}
    Equations~\eqref{eqn:hausdorff_upper_bound} and~\eqref{eqn:hausdorff_lower_bound} prove that $\dH(A,B) \leq t$ (and also $\dDH(B,A) \leq t$) if and only if~$\Phi$ is true.
\end{proof}

In the proof of \cref{thm:hausdorff_strict_uetr_hard}, we could choose~$N':=N+1$ instead of~$N$ in \cref{eqn:definition_of_N}.
Then in the case that~$\Phi$ is false, the Hausdorff distance~$\dH(A,B)$ is at least
\[
    2^{2^{N+1}} r
    > 2^{2^{N+1}-2^N} m
    = 2^{2^N} m
    = 2^{2^N} t
    \text{.}
\]
Note that the number of free variables in the formulas describing the resulting sets~$A$ and~$B$ equals $n + m + N' +1 = \Theta(N)$.
We created a gap of size~$\cramped{2^{2^{\Theta(N)}}}$.
This implies the following inapproximability result.

\approx*

Another interesting observation is that we can restrict the sets~$A$ and~$B$ to be described by (relatively) simple formulas:

\simplicity*

\begin{proof}
    Taking the formula $\psi$ in the proof of \cref{thm:hausdorff_strict_uetr_hard}, we apply \mbox{\cref{lem:tseitin}\ref{itm:tseitin_quadratic_equations}} to obtain an equivalent new formula~$\psi'$ (with additional existentially quantified variables) which is a conjunction of quadratic polynomial equations.
    Then set~$A$ can be defined using~$\psi'$ instead of~$\psi$.
    Set~$B$ can be trivially described in the desired form.
    This shows statement \ref{itm:simplicity_quadratic_equations}.

    For \ref{itm:simplicity_single_equation}, we modify the above procedure by applying \mbox{\cref{lem:tseitin}\ref{itm:tseitin_single_equation}} to~$\psi$ to obtain an equivalent formula which is a single polynomial of degree at most four.
\end{proof}

\section{\texorpdfstring{\UER[<]}{UE<R}-Membership of \Hausdorff}
\label{sec:hausdorff_membership}

This section is devoted to show the following \lcnamecref{thm:hausdorff_membership}.

\begin{theorem}
    \label{thm:hausdorff_membership}
    \Hausdorff and \DirectedHausdorff are contained in \UER[<].
\end{theorem}

Note that \UER-membership is already shown by Dobbins et al.~\cite{Dobbins2018_AreaUniversality}.
The remainder of this section deals with reformulating a given \Hausdorff instance into a \StrictUETR instance, thereby proving \UER[<]-membership.

Let~$(A,B,t)$ be a \Hausdorff instance, where $A = \{x \in \R^n \mid \varphi_A(x)\}$ and $B = \{x \in \R^n \mid \varphi_B(x)\}$ are described by quantifier-free formulas~$\varphi_A$ and~$\varphi_B$ with~$n$ free variables each.
For simplicity, we only consider the directed Hausdorff distance here, namely the question whether
\[
    \dDH(A,B) :=
    \adjustlimits \sup_{a \in A} \inf_{b \in B} \, \norm{a - b} \stackrel{?}{\leq} t
    \text{.}
\]
It is obvious, that $\dH(A,B) \leq t$ if and only if $\dDH(A,B) \leq t$ and $\dDH(B,A) \leq t$.
So if we can formulate the decision problem for the directed Hausdorff distance as a \StrictUETR instance, their conjunction is a formula for the general \Hausdorff problem.
Assuming that no variable name appears in both operands of this conjunction, this formula can be converted into prenex normal form by just moving the quantifiers to the front.

From the definition we get that $\dDH(A,B) \leq t$ is equivalent to
\begin{equation}
    \label{eqn:hausdorff_as_uetr}
    \forall \eps > 0, a \in A \sepdot
    \exists b \in B :
    \norm{a - b}^2 < t^2 + \eps
    \text{.}
\end{equation}
Let us remark that introducing the real variable~$\eps$ is necessary to also consider the points in the closure of~$B$.
Moreover, we work with the squared distance between~$a$ and~$b$, because~$\norm{a - b}$ is the square root of a polynomial.

Below we transform formula \eqref{eqn:hausdorff_as_uetr} in multiple technical steps into a form that allows us to apply a recent theorem by D'Costa et al.~\cite{DCosta2021_EscapeProblem} such that \UER[<]-membership follows.
Before we do so, we state a few helpful lemmas.
These allow us to consider some of the intermediate steps in isolation, thereby simplifying the needed notation.
Also, \cref{lem:universal_epsiolon_to_special_uetr} below is used again in \cref{sec:exotic_quantifiers}.

The first \lcnamecref{lem:special_uetr_to_strict_uetr} allows us to transform a \UETR instance of special structure into an equivalent \StrictUETR instance:

\begin{lemma}
    \label{lem:special_uetr_to_strict_uetr}
    Given a \UETR instance
    \[
        \Phi \wdequiv
        \forall X \in [-1,1]^n \sepdot
        \exists Y \in [-1,1]^m :
        \varphi_<(X,Y) \lor H(X,Y) = 0
        \text{,}
    \]
    where~$\varphi_<(X,Y) \in \QFF[<]$ and~$H : [-1,1]^{n+m} \to \R$ is a polynomial.
    Then we can compute in polynomial time an equivalent \StrictUETR instance.
\end{lemma}

\begin{proof}
    We first prove that there exists an integer~$N \leq \poly(\abs{\Phi})$, such that the \StrictUETR instance
    \[
        \Psi \wdequiv
        \forall X \in [-1,1]^n \sepdot
        \exists Y \in [-1,1]^m :
        \varphi_<(X,Y) \lor H(X,Y)^2 < 2^{-2^N}
    \]
    is equivalent to~$\Phi$.
    In a second step, we construct~$\cramped{2^{-2^N}}$ inside the formula.

    The direction~$\Phi \implies \Psi$ is trivially true for any~$N \in \N$.
    To prove the other direction, we show its contraposition $\neg\Phi \implies \neg\Psi$.
    Assume that
    \[
        \neg\Phi \wequiv
        \exists X \in [-1,1]^n \sepdot
        \forall Y \in [-1,1]^m :
        \neg\varphi_<(X,Y) \land H(X,Y)^2 > 0
    \]
    is true.
    Hence for at least one fixed~$x \in [-1,1]^n$ we obtain a polynomial~$H(x,Y)^2$ that is positive everywhere on~$[-1,1]^m$ (the fixed~$x$ values are real coefficients for the variables~$Y$).
    Because~$[-1,1]^m$ is compact and because polynomials are continuous, $H(x,Y)^2$ attains its minimum over~$[-1,1]^m$ and it follows that
    \begin{equation}
        \label{eqn:polynomial_at_least_eps}
        \exists \eps > 0 \sepdot
        \exists X \in [-1,1]^n \sepdot
        \forall Y \in [-1,1]^m :
        \neg\varphi_<(X,Y) \land H(X,Y)^2 \geq \eps
    \end{equation}
    is true.
    Let~$L$ be the length of the quantifier-free part in \eqref{eqn:polynomial_at_least_eps}.
    By \cref{lem:lower_bound_leading_eps} there is a constant~$\beta \in \R$ such that~$\exists \eps > 0$ in~\eqref{eqn:polynomial_at_least_eps} can be strengthened to $\exists\eps \geq \cramped{2^{-L^{\beta^{4}nm}}}$.
    Now choose~$N = \bigl\lceil \beta^{4}nm \bigr\rceil$ to be the smallest integer satisfying~$\cramped{2^{-2^N}} < \cramped{2^{-L^{\beta^{4}nm}}}$.
    Note that $N \leq \poly(n,m, \log L)$, so it is polynomial in the input size.
    Plugging in the lower bound on~$\eps$, we get that~$\neg\Phi$ is equivalent to
    \[
        \exists X \in [-1,1]^n \sepdot
        \forall Y \in [-1,1]^m :
        \neg\varphi_<(X,Y) \land H(X,Y)^2 \geq 2^{-2^N}
        \text{,}
    \]
    which is exactly~$\neg\Psi$.
    We conclude that~$\Phi \equiv \Psi$ as claimed.
    
    To construct a \StrictUETR instance from~$\Psi$, we need to express~$\cramped{2^{-2^N}}$ inside the formula.
    To this end, introduce~$N + 1$ new variables $U = (U_0, \ldots, U_N) \in [-1,1]^{N+1}$ and let~$\chi(U)$ be formula~\eqref{eqn:double_exponetial_small_value}.
    Recall that by~\cref{lem:double_exponential_small_value},~$\chi(u)$ is true if and only if~$u_i = \cramped{2^{-2^i}}$.
    Including~$\chi(U)$ into our formula, we conclude that
    \begin{equation}
        \label{eqn:bounded_strict_uetr}
        \forall X, U \in [-1,1]^{n + N + 1} \sepdot
        \exists Y \in [-1,1]^m :
        \neg\chi(U) \lor \neg\varphi_<(X,Y) \lor H(X,Y)^2 < U_N
    \end{equation}
    is equivalent to~$\Phi$.
    
    We arrived at a sentence where all variables are restricted to~$[-1,1]$ and in which all atoms are strict.
    At this point we use a recent result by D'Costa et al.~\cite{DCosta2021_EscapeProblem}.
    They show that it is \EUR[\leq]-complete to decide a sentence of the form
    \[
        \exists X \in [-1,1]^n \sepdot
        \forall Y \in [-1,1]^m :
        \varphi_\leq(X,Y)
    \]
    with~$\varphi_\leq \in \QFF[\leq]$.
    Because~$\EUR[\leq] = \co\UER[<]$, deciding the complements of these sentences, i.e., sentences of the form
    \[
        \forall X \in [-1,1]^n \sepdot
        \exists Y \in [-1,1]^m :
        \varphi_<(X,Y)
    \]
    with~$\varphi_< \in \QFF[<]$ is \UER[<]-complete.
    Sentence~\eqref{eqn:bounded_strict_uetr} is of this form.
    Thus there is a polynomial-time reduction to an equivalent \StrictUETR instance.
\end{proof}

The next lemma establishes an upper bound on the value of a polynomial over a compact domain.

\begin{lemma}
    \label{lem:polynomial_over_compact_domain}
    Let~$P : \R^n \to \R$ be a polynomial, $N$ be an integer and~$C := \cramped{2^{2^N}}$.
    Then we can compute in polynomial time an integer~$K \leq \poly(\abs{P}, N, n)$ such that for~$E := \cramped{2^{2^K}}$ and all~$x \in [-C,C]^n$ it holds that $\abs{P(x)} \leq E$.
\end{lemma}

\begin{proof}
    Because~$P$ is a polynomial,~$\abs{P}$ is continuous and therefore~$\abs{P}$ attains its maximum over any compact domain.
    We conclude that
    \[
        \exists E \in \R \sepdot
        \forall X \in [-C,C]^n :
        \abs{P(X)} \leq E
    \]
    is true.
    Note that, strictly speaking, we may not use $\abs{\cdot}$ inside the formula.
    However, $\abs{P(X)} \leq E$ is equivalent to $P(X) \leq E \land -P(X) \leq E$.
    
    To obtain an upper bound on~$E$, we first need to encode~$C$ inside the formula.
    We introduce~$N + 1$ new variables $U = \{U_0, \ldots, U_N\}$ and let~$\chi(U)$ be formula~\eqref{eqn:double_exponetial_small_value}.
    Recall that by \cref{lem:double_exponential_small_value}, $\chi(u)$ is true if and only if~$u_i = \cramped{2^{-2^i}}$.
    Now we can rewrite the above formula equivalently as
    \[
        \exists E \in \R \sepdot
        \forall X \in \R^n \sepdot
        \exists U \in \R^{N+1} :
        \bigwedge\limits_{i = 1}^n \abs{X_i U_N} \leq 1 \implies \abs{P(X)} \leq E
        \text{.}
    \]
    Let~$\varphi(E)$ be the subformula following the quantification of~$E$ (starting from $\forall$) and~$L := \abs{\varphi(E)}$.
    Applying quantifier elimination (\cref{cor:single_exponential_quantifier_elimination}) to $\varphi(E)$, we obtain a constant~$\alpha \in \R$ and an equivalent, quantifier-free formula~$\psi(E)$ of length
    \[
        \abs{\psi(E)} \leq L^{\alpha^3n(N+1)}
        \text{.}
    \]
    The ball theorem (\cref{thm:ball_theorem}) applied to $\exists E \in \R : \psi(E)$ now yields an upper bound for~$E$:
    \[
        E \leq
        2^{\abs{\psi}^8} \leq
        2^{L^{8 \alpha^3 n (N+1)}} =
        2^{2^{8 \log(\abs{\psi}) \alpha^3 n (N+1)}}
    \]
    Choose~$K = \bigl\lceil 8 \log(\abs{\psi}) \alpha^3 n (N+1) \bigr\rceil$.
    Obviously, $K \leq \poly(\log \abs{\psi}, N, n)$.
    Since~$\abs{\psi} \leq \poly(\abs{P}, n, N)$, the claim follows.
\end{proof}

Above \lcnamecref{lem:polynomial_over_compact_domain} is used to prove the following \lcnamecref{lem:universal_epsiolon_to_special_uetr} that allows us to transform some more general \UETR instances into equivalent \StrictUETR instances.

\begin{lemma}
    \label{lem:universal_epsiolon_to_special_uetr}
    Given a \UETR instance
    \[
        \forall \eps > 0, X \in \R^n \sepdot
        \exists Y \in \R^m :
        F(X)^2 > 0 \lor \bigl(
            G(Y) = 0 \land P(X,Y) < \eps
        \bigr)
        \text{,}
    \]
    where~$F : \R^n \to \R$, $G : \R^m \to \R$ and $P : \R^{n+m} \to \R$ are polynomials.
    Then we can compute in polynomial time an equivalent \StrictUETR instance.
\end{lemma}

\begin{proof}
    Via a series of manipulations, we transform the given sentence into an equivalent \UETR instance that has the form required by \cref{lem:special_uetr_to_strict_uetr}.
    The first step is to move the condition that~$\eps > 0$ into the formula.
    We obtain an equivalent sentence
    \begin{align*}
        \forall \eps \in \R, X \in \R^n :~
        &(\eps > 0) \implies \\
        &\bigl(
            \exists Y \in \R^n :
            F(X)^2 > 0 \lor (G(Y) = 0 \land P(X,Y) < \eps)
        \bigr)
        \text{.}
    \end{align*}
    Now we observe that~$\eps > 0$ is equivalent to $\exists \delta \in \R : \delta^2\eps - 1 = 0$.
    Incorporating this yields an equivalent sentence
    \begin{align*}
        \forall \eps \in \R, X \in \R^n :~
        &(\exists \delta \in \R : \delta^2\eps - 1 = 0) \implies \\
        &\bigl(
            \exists Y \in \R^m :
            F(X)^2 > 0 \lor (G(Y) = 0 \land P(X,Y) < \eps)
        \bigr)
        \text{.}
    \end{align*}
    Rewriting the implication~$A \implies B$ as~$\neg A \lor B$ turns the existential quantifier in front of~$\delta$ into a universal quantifier.
    Furthermore, we replace $\neg(\delta^2\eps - 1 = 0)$ by the equivalent $(\delta^2\eps - 1)^2 > 0$.
    We get an equivalent sentence
    \begin{align*}
        \forall \eps \in \R, X \in \R^n :~
        &\bigl(
            \forall \delta \in \R : (\delta^2\eps - 1)^2 > 0
        \bigr)~\lor \\
        &\bigl(
             \exists Y \in \R^m :F(X)^2 > 0 \lor (G(Y) = 0 \land P(X,Y) < \eps)
        \bigr)
        \text{.}
    \end{align*}
    Moving all quantifiers to the front, turns this into an equivalent prenex normal form
    \begin{align*}
        &\forall \eps \in \R, \delta \in \R, X \in \R^n \sepdot
        \exists Y \in \R^m :~\\
        &(\delta^2\eps + 1)^2 > 0 \lor F(X)^2 > 0 \lor
        \bigl(
            G(Y) = 0 \land P(X,Y) < \eps
        \bigr)
        \text{.}
    \end{align*}
    This sentence is \UStrict, so \cref{lem:bounding_universal_range,lem:bounding_existential_range} are applicable.
    Thus, there are two integers~$N, M$ bounded by a polynomial in the length of the sentence, such that for $C := \cramped{2^{2^N}}$ and $D := \cramped{2^{2^M}}$ all universally quantified variables can be restricted to~$[-C,C]$ and all existentially quantified variables can be restricted to~$[-D,D]$.
    We obtain another equivalent sentence
    \begin{align*}
        &\forall \eps \in [-C,C], \delta \in [-C,C], X \in [-C,C]^n \sepdot
        \exists Y \in [-D,D]^m :~\\
        &(\delta^2\eps + 1)^2 > 0 \lor F(X)^2 > 0 \lor
        \bigl(
            G(Y) = 0 \land P(X,Y) < \eps
        \bigr)
        \text{.}
    \end{align*}
    In the next step, we replace the strict inequality~$P(X,Y) < \eps$ by the non-strict inequality $P(X,Y) \leq \eps$.
    For this step, we exploit the fact that a continuous function over a compact domain attains its minimum and maximum.
    Consequently, $\forall \eps >0, X\in [-C,C] \sepdot \exists Y\in [-D,D]: P(X,Y)<\eps$ is (true if and only if 
    $\max_{X\in[-C,C]}\min_{Y\in [-D,D]} P(X,Y)\leq 0$ and thus) 
    equivalent to $\forall \eps >0, X\in [-C,C]\sepdot Y\in [-D,D] : P(X,Y) \leq \eps$.
    We obtain the equivalent sentence
    \begin{align*}
        &\forall \eps \in [-C,C], \delta \in [-C,C], X \in [-C,C]^n \sepdot
        \exists Y \in [-D,D]^m :~\\
        &(\delta^2\eps + 1)^2 > 0 \lor F(X)^2 > 0 \lor
        \bigl(
            G(Y) = 0 \land P(X,Y) \leq \eps
        \bigr)
        \text{.}
    \end{align*}
    Going one step further, we now want to express~$P(X,Y) \leq \eps$ as a polynomial equation.
    To this end, we replace it by the equivalent formula $\exists B \in \R : P(X,Y) - \eps + B^2 = 0$.
    By \cref{lem:polynomial_over_compact_domain}, we can also bound the range over which~$B$ is quantified:
    We can compute in polynomial time an integer $K \leq \poly(\abs{P}, \max\{N,M\}, n + m + 1)$ such that~$\abs{B} \leq E := \cramped{2^{2^K}}$.
    We get another equivalent sentence
    \begin{align*}
        &\forall \eps \in [-C,C], \delta \in [-C,C], X \in [-C,C]^n \sepdot
        \exists Y \in [-D,D]^m, B \in [E,E] :~\\
        &(\delta^2\eps + 1)^2 > 0 \lor F(X)^2 > 0 \lor
        \bigl(
            G(Y) = 0 \land P(X,Y) - \eps + B^2 = 0
        \bigr)
        \text{.}
    \end{align*}
    At this point we define
    \begin{align*}
        \varphi_<(\eps, \delta, X)\ &:=\ (\delta^2\eps + 1)^2 \lor F(X)^2 > 0
        \quad \text{and} \\
        H(X, Y, \eps, B)\ &:=\ G(Y)^2 + (P(X,Y) - \eps + B^2)^2
        \text{.}
    \end{align*}
    Note that~$\varphi_< \in \QFF[<]$.
    We use these to get the equivalent sentence
    \begin{align*}
        &\forall \eps \in [-C,C], \delta \in [-C,C], X \in [-C,C]^n \sepdot
        \exists Y \in [-D,D]^m, B \in [E,E] :~\\
        &\varphi_<(\eps, \delta, X) \lor H(X, Y, \eps, B) = 0
        \text{.}
    \end{align*}
    The last step is to scale all variables to be in the interval~$[-1,1]$.
    For this, let $S := \max\{N, M, K\}$ and introduce~$S + 1$ new variables $U = \{U_0, \ldots, U_S\}$.
    Furthermore, let~$\chi(U)$ be formula~\eqref{eqn:double_exponetial_small_value}.
    Recall that by \cref{lem:double_exponential_small_value}, $\chi(u)$ is true if and only if~$u_i = \cramped{2^{-2^i}}$.
    We can rewrite our sentence to the equivalent sentence
    \begin{align*}
        &\forall \eps, \delta, X, U \in [-1,1]^{1 + 1 + n + S + 1} \sepdot
        \exists Y, B \in [-1,1]^{m + 1} :~\\
        &\neg\chi(U) \lor
        U_S^d \cdot \varphi_<\left(
            \frac{\eps}{U_S}, \frac{\delta}{U_S}, \frac{X}{U_S}
        \right) \lor
        U_S^d \cdot H\left(
            \frac{X}{U_S}, \frac{Y}{U_S}, \frac{\eps}{U_S}, \frac{B}{U_S}
        \right) = 0
        \text{.}
    \end{align*}
    Here~$d$ is the maximum degree of any polynomial in~$\varphi_<$ and~$H$.
    By~$X/U_S$ we denote that every variable~$X_i$ is replaced by~$X_i/U_S$.
    Multiplying by~$U_S^d$ makes sure that each atom remains a polynomial.
    What we obtained is a \UETR instance that has the form required by \cref{lem:special_uetr_to_strict_uetr} (note that~$\neg\chi(U)$ is strict).
    Therefore, we can transform this into an equivalent \StrictUETR instance in polynomial time.
\end{proof}

Now we finally have all the needed tools to prove \cref{thm:hausdorff_membership} which states that the \Hausdorff problem is contained in \UER[<].
We do this by transforming formula~\eqref{eqn:hausdorff_as_uetr} into the form required by \cref{lem:universal_epsiolon_to_special_uetr}.
This then yields an equivalent \StrictUETR instance, thus proving \UER[<]-membership.

\begin{proof}[Proof of \cref{thm:hausdorff_membership}]
    Recall that $\dDH(A,B) \leq t$ is equivalent to~\eqref{eqn:hausdorff_as_uetr}, which is
    \[
        \forall \eps > 0, a \in A \sepdot
        \exists b \in B :
        \norm{a - b}^2 < t^2 + \eps
        \text{.}
    \]
    In a first step, we resolve the shorthand notations $a \in A$ and $b\in B$ and we obtain
    \[
        \forall \eps > 0, a \in \R^n :
        \varphi_A(a) \implies
        \bigl(
            \exists b \in \R^n : \varphi_B(b) \land \norm{a - b}^2 < t^2 + \eps
        \bigr)
        \text{.}
    \]
    Next, we consider the (quantifier-free) formulas~$\varphi_A(a)$ and~$\varphi_B(b)$.
    Using \cref{lem:tseitin}, we obtain in polynomial time two integers~$k, \ell$ and two polynomials $F_A : \R^{n+k} \to \R$ and $F_B : \R^{n+\ell} \to \R$, such that $\varphi_A(a)$ is equivalent to $\exists U_a \in \R^k : F_A(a,U_a) = 0$ and similarly~$\varphi_B(b)$ is equivalent to $\exists U_b \in \R^\ell : F_B(b,U_b) = 0$.
    This yields the equivalent sentence
    \begin{align*}
        \forall \eps > 0, a \in \R^n :~
        &\bigl(
            \exists U_a \in \R^k : F_A(a, U_a) = 0
        \bigr) \implies \\
        &\bigl(
            \exists b \in \R^n, U_b \in \R^\ell : F_B(b, U_b) = 0 \land \norm{a - b}^2 < t^2 + \eps
        \bigr)
        \text{.}
    \end{align*}
    Rewriting the implication~$X \implies Y$ as~$\neg X \lor Y$ changes the existential quantifier in front of~$U_a$ into a universal quantifier, which we can move to the front.
    Also, the equation gets negated.
    Substituting~$\neg(F(a,U_a) = 0)$ by $F(a,U_a)^2 > 0$, we get the equivalent sentence
    \begin{align*}
        \forall \eps > 0, a \in \R^n, U_a \in \R^k :~
        &F_A(a, U_a)^2 > 0~\lor \\
        &\bigl(
            \exists b \in \R^n, U_b \in \R^\ell : F_B(b, U_b) = 0 \land \norm{a - b}^2 < t^2 + \eps
        \bigr)
        \text{.}
    \end{align*}
    Lastly, we move the existential quantifier after the universal one and get an equivalent sentence
    \begin{align*}
        &\forall \eps > 0, a \in \R^n, U_a \in \R^k \sepdot
        \exists b \in \R^n, U_b \in \R^\ell : \\
        & F_A(a, U_a)^2 > 0 \lor
        \bigl(
            F_B(b, U_b) = 0 \land \norm{a - b}^2 < t^2 + \eps
        \bigr)
    \end{align*}
    in prenex normal form.
    This sentence has the form required by \cref{lem:special_uetr_to_strict_uetr}, concluding the proof.
\end{proof}

\section{Exotic Quantifiers}
\label{sec:exotic_quantifiers}

In this section, we show an interesting connection of the complexity class $\UER[<]$ to a complexity class introduced by Bürgisser and Cucker when studying the computational complexity of many basic problems regarding semi-algebraic sets~\cite{Burgisser2009_ExoticQuantifiers}.
While they mainly work in the BSS-model, they also consider some problems in the bit-model of computation (which we use throughout this paper).
There, among others, the complexity classes \ER, \UR and \UER appear under the names $\text{\textsf{BP}}^0(\exists)$, $\text{\textsf{BP}}^0(\forall)$, and $\text{\textsf{BP}}^0(\forall\exists)$, respectively.
(Here the \textsf{BP} stands for \enquote{binary part} and the superscript~$0$ denotes that there are no machine constants in the BSS machine.)

They notice that the computational complexity of some natural problems defies to be classified into this hierarchy of complexity classes. 
In this paper, we have a similar situation because we prove that the \Hausdorff problem is complete for a the class \UER[<], supposedly between \ER/\UR and \UER.
See also the paper by D'Costa et al.~\cite{DCosta2021_EscapeProblem}, where they show that their \emph{escape problem} is \EUR[\leq]-complete, though supposedly between \ER/\UR and \EUR.

The approach in this paper and in~\cite{DCosta2021_EscapeProblem} is to make syntactic restrictions to the quantifier-free parts of the sentences (i.e., allowing only strict inequalities as in \StrictUETR).
This defines new decision problems that in turn are used to define new complexity classes like \UER[<].
Bürgisser and Cucker take a different approach.
They define new quantifiers, called \emph{exotic quantifiers},  that make a topological restriction on the sentences~\cite{Burgisser2009_ExoticQuantifiers}. Two of them are highly related to our work:
\begin{align*}
    \forall^* X \in \R^n : \varphi(X)
    &\wdequiv
    \forall \eps > 0, X \in \R^n \sepdot
    \exists \til{X} \in \R^n :
    \norm{X - \til{X}}^2 < \eps \land \varphi(\til{X}) \\
    \exists^* X \in \R^n : \varphi(X)
    &\wdequiv
    \exists \eps > 0, X \in \R^n \sepdot
    \forall \til{X} \in \R^n :
    \norm{X - \til{X}}^2 < \eps \implies \varphi(\til{X})
\end{align*}
Intuitively, $\forall^* X \in \R^n : \varphi(X)$ means that~$\varphi(x)$ does not need to be true for all~$x \in \R^n$ but just for all $x \in D$, where~$D$ is some dense subset of~$\R^n$.
Conversely, $\exists^* X \in \R^n : \varphi(X)$ means that there must be an~$x \in \R^n$ and some radius~$r > 0$, such that for all~$\til{x} \in B_n(x, r)$ it holds that $\varphi(\til{x})$, where $B_n(x, r)$ denotes the $n$-dimensional open ball of radius~$r$ centered at~$x$.

As one would expect, it holds that
\begin{align*}
    \neg\forall^* X \in \R^n : \varphi(X)
    &\wequiv
    \exists^* X \in \R^n : \neg\varphi(X)
    \quad \text{and} \\
    \neg\exists^* X \in \R^n : \varphi(X)
    &\wequiv
    \forall^* X \in \R^n : \neg\varphi(X)
    \text{.}
\end{align*}

\cref{thm:exotic_uetr}, the main theorem of this section, establishes a relation between the two approaches described above.
As it turns out, the topological restrictions on the formulas by Bürgisser and Cucker are equivalent to the syntactical restrictions done for example in this paper.
For this purpose, let \ExoticUETR denote the decision problem whether a sentence of the form
\[
    \forall^* X \in \R^n \sepdot
    \exists Y \in \R^m : 
    \varphi(X,Y)
\]
with quantifier-free~$\varphi$ is true.
Further, we define the complexity class $\forall^*\exists\R$ to contain all problems that polynomial-time many-one reduce to \ExoticUETR.
We show that the complexity classes $\forall^*\exists\R$ and $\UER[<]$ coincide.

\exoticuetr*

\begin{proof}
    The \UER[<]-hardness of \ExoticUETR follows directly from \cref{thm:open_ball_counterexamples}.
    For a given  \StrictUETR instance
    \[
        \Phi \wdequiv
        \forall X \in \R^n \sepdot
        \exists Y \in \R^m :
        \varphi_<(X,Y)
    \]
    with~$\varphi_< \in \QFF[<]$, \cref{thm:open_ball_counterexamples} allows to compute in polynomial time  an equivalent \UETR instance
    \[
        \Psi \wdequiv
        \forall X \in \R^k \sepdot
        \exists Y \in \R^\ell :
        \psi(X,Y)
    \]
    with~$\psi \in \QFF$.
    Recall that on the one hand, if~$\Psi$ is false, then the set of counterexamples~$\counterexample{\Psi}$ contains an open ball.
    On the other hand, if~$\Psi$ is true, then~$\counterexample{\Psi} = \emptyset$.
    Therefore,~$\Psi$ is true if and only if it is true for a dense subset of~$\R^k$.
    It follows that the $\forall$-quantifier can be replaced by the exotic $\forall^*$-quantifier in~$\Psi$ and we get
    \[
        \Phi \wequiv \Psi \wequiv
        \forall^* X \in \R^k \sepdot
        \exists Y \in \R^\ell :
        \psi(X,Y)
        \text{.}
    \]
    Consequently, \ExoticUETR is  \UER[<]-hard.
    
    To prove that \ExoticUETR is contained in \UER[<], we transform
    \[
        \forall^* X \in \R^n \sepdot
        \exists Y \in \R^m :
        \varphi(X,Y)
    \]
    in polynomial time into an equivalent sentence of the form required by \cref{lem:universal_epsiolon_to_special_uetr}.
    This \lcnamecref{lem:universal_epsiolon_to_special_uetr} allows us to construct an equivalent \StrictUETR instance in polynomial time, thereby proving \UER[<]-membership.
    We start by expressing the exotic quantifier~$\forall^*$ in terms of classical quantifiers~$\forall$ and~$\exists$, obtaining an equivalent sentence
    \[
        \forall \eps > 0, X \in \R^n \sepdot
        \exists X_0 \in \R^n, Y \in \R^m :
        \norm{X - X_0}^2 < \eps \land \varphi(X_0, Y)
        \text{.}
    \]
    By \cref{lem:tseitin} we can compute in polynomial time an integer~$k$ and a polynomial $G : \R^{n + m + k} \to \R$ such that~$\varphi(X_0,Y)$ is equivalent to $\exists U \in \R^k : G(X_0, Y, U) = 0$.
    Plugging this into above sentence, we get another equivalent sentence
    \[
        \forall \eps > 0, X \in \R^n \sepdot
        \exists X_0 \in \R^n, Y \in \R^m, U \in \R^k :
        \norm{X - X_0}^2 < \eps \land G(X_0, Y, U) = 0
        \text{.}
    \]
    This has the form required by \cref{lem:universal_epsiolon_to_special_uetr}.
    Hence, \UER[<]-membership follows.
\end{proof}

\subsection{Applications to the \Hausdorff Problem}
\label{sec:erd_hausdorff}

We now use our insights to establish the exact computational complexity of \problemname{Euclidean Relative Denseness (ERD)} which was left as an open problem by Bürgisser and Cucker~\cite{Burgisser2009_ExoticQuantifiers}.
In ERD, we are given two semi-algebraic sets~$A$ and~$B$ and wonder whether~$A$ is contained in the closure of~$B$, which is denoted by~$\overline{B}$.
Note that \problemname{ERD} is equivalent to deciding whether~$\dDH(A,B) = 0$.
Bürgisser and Cucker show:

\begin{theorem}[{\cite[Corollary 5.6]{Burgisser2009_ExoticQuantifiers}}]
    \label{thm:erd_bc}
    \problemname{ERD} is in \UER and $\forall^*\exists\R$-hard.
\end{theorem}

They prove this in the BSS-model, but the same proof also works in the bit-model.
Building upon our insights, we are able to determine the exact computational complexity of \problemname{ERD} (in the bit-model).

\begin{theorem}
    \problemname{ERD} is \UER[<]-complete.
\end{theorem}

\begin{proof}
    By \cref{thm:exotic_uetr} and $\forall^*\exists\R$-hardness from \cref{thm:erd_bc}, it follows that \problemname{ERD} is \UER[<]-hard.
    Further, \cref{thm:hausdorff_membership} implies that \problemname{ERD} is contained in \UER[<].
    Consequently, \problemname{ERD} is \UER[<]-complete.
\end{proof}

Moreover, we remark that  \UER[<]-hardness of \problemname{ERD} implies \UER[<]-hardness of the general directed \Hausdorff problem (for any distance~$t \geq 0$):
Given an instance~$A,B \subseteq \R^n$ of \problemname{ERD}, we define $A' := (A,0) \subseteq \R^{n+1}$ and $B' := (B,1) \subseteq \R^{n+1}$.
Then $\dDH(A,B) = 0$ if and only if~$\dDH(A',B') \leq 1$. 

\section{Open Problems}
\label{sec:open_problems}

We showed that the \Hausdorff problem is \UER[<] complete.
We conclude the paper with a list of interesting open questions:\medskip

While $\ER = \exists_{\leq}\R  = \exists_<\R$ and $\UR = \forall_{\leq}\R = \forall_{<}\R$ are known \cite{Schaefer2017_FixedPointsNash}, similar results are unknown for higher levels of the hierarchy.
In this context, we are particularly interested in the following question:

\begin{openproblem}
    \label{op:complexity_classes_equal}
    Are the two complexity classes \UER[<] and \UER actually the same?
\end{openproblem}

Regarding \cref{op:complexity_classes_equal}, we are not aware of any algorithms that are more efficient if the atoms are restricted to strict polynomial inequalities only.
An answer to this question is interesting in its own right.
If the two classes are indeed different, this would imply $\NP \neq \PSPACE$, so we do not expect such a proof any time soon.

\begin{openproblem}
    \label{op:forall_strict_uetr}
    What is the computational complexity of deciding \UStrict~\UETR instances?
\end{openproblem}

By definition, deciding \UStrict~\UETR instances is \UER[<]-hard and in \UER.
The most important tools for showing \UER[<]-membership in this paper and also in~\cite{DCosta2021_EscapeProblem} are the range restrictions for the quantified variables (\cref{lem:bounding_universal_range,lem:bounding_existential_range}).
These are still applicable to \UStrict formulas.
However, this seems not enough to transform an arbitrary \UStrict~\UETR instance efficiently into an equivalent \StrictUETR instance, in which all atoms are not negated and are strict inequalities.
If this can be done, then this would significantly simplify the proof of \cref{lem:universal_epsiolon_to_special_uetr} (and therefore \cref{thm:hausdorff_membership} proving that \Hausdorff is in \UER[<]).\medskip

Our third open problem concerns even more restricted versions of the \Hausdorff problem that remain hard.
As already noted in \cref{cor:simple_formulas}, \Hausdorff remains \UER[<]-hard if all atoms in the formulas describing~$A$ and~$B$ are quadratic equations.

\begin{openproblem}
    \label{op:simlest_hard_setting}
    What is the most restricted version of the \Hausdorff problem that is still \UER[<]-hard?
\end{openproblem}

There are several directions to explore on \cref{op:simlest_hard_setting}.
Identifying meaningful restrictions on the sets~$A$ and~$B$ might lead to an easier problem.
For example, it is \UR-complete to decide if closed semi-algebraic sets~$A,B$ have Hausdorff distance exactly zero (they do if and only if~$A = B$):

\begin{theorem}
    Deciding if two semi-algebraic sets are equal is \UR-complete.
\end{theorem}

\begin{proof}
    Given quantifier-free formulas $\varphi_A(X)$ and~$\varphi_B(X)$, it holds that $A = B$ if and only if $\forall X \in \R^n : \varphi_A(X) \iff \varphi_B(X)$.
    This shows \UR-membership.
    
    To prove \UR-hardness, note that the sentence $\forall X \in \R^n : \varphi(X)$ is equivalent to deciding whether $\{x \in \R^n \mid \varphi(x)\} \stackrel{?}{=} \R^n$.
\end{proof}

\bibliographystyle{plainurl}
\bibliography{references}

\end{document}